\documentclass[journal,twoside,web]{ieeecolor}
\usepackage{generic}
\usepackage{cite}
\usepackage{amsmath,amssymb,amsfonts}
\usepackage{algorithmic}
\usepackage{graphicx}
\usepackage{algorithm,algorithmic}
\usepackage{hyperref}
\hypersetup{hidelinks=true}
\usepackage{textcomp}
\usepackage{pgfplots}
\usepgfplotslibrary{fillbetween}
\usetikzlibrary {arrows.meta}
\pgfplotsset{compat=1.16}
\usepackage{tikz-3dplot}
\usepackage[acronym]{glossaries}

\newtheorem{thm}{Theorem}
\newtheorem{assum}{Assumption}
\newtheorem{remark}{Remark}
\newtheorem{lemma}{Lemma}
\newtheorem{problem}{Problem}
\newtheorem{cor}{Corollary}

\newcommand{\estim}[1]{\hat{\mathbf{U}}_{#1}}
\newcommand{\estimST}[1]{\hat{U}_{#1}}
\newcommand{\estimGGD}[2]{\mathbf{\Gamma}_{#2}}
\newcommand{\estimGGDST}[2]{\Gamma_{#2}}
\newcommand{\subs}[1]{\mathbf{U}_{#1}}
\newcommand{\subsST}[1]{U_{#1}}
\newcommand{\traj}[2]{u_{#2}}
\newcommand{\ntraj}[2]{\bar{u}_{#2}}
\def\sysorder{k}
\def\windowlength{T}
\def\gradnoise{N_t}

\definecolor{myBlue}{rgb}{0,0.263,0.576}
\definecolor{myRed}{RGB}{216,27,96}
\definecolor{myGreen}{RGB}{255,193,7}

\newacronym{GREAT}{GREAT}{\textit{Grassmannian Recursive Algorithm for Tracking}}
\newacronym{SVD}{SVD}{singular value decomposition}
\newacronym{LTI}{LTI}{Linear Time-Invariant}
\newacronym{LTV}{LTV}{Linear Time-Varying}
\newacronym{PAST}{PAST}{Projection Approximation Subspace Tracking}
\newcommand\Grass[2]{\operatorname{Gr}(#1, #2)}

\def\mytitle{GREAT: Grassmannian REcursive Algorithm for Tracking \& Online System Identification}

\def\BibTeX{{\rm B\kern-.05em{\sc i\kern-.025em b}\kern-.08em
    T\kern-.1667em\lower.7ex\hbox{E}\kern-.125emX}}
\begin{document}
\title{\mytitle}

\author{Andr\'as Sasfi, Alberto Padoan, \IEEEmembership{Member, IEEE}, Ivan Markovsky, \IEEEmembership{Member, IEEE}, Florian D\"orfler, \IEEEmembership{Senior Member, IEEE}
\thanks{A. Sasfi, A. Padoan, and F. D\"orfler are with the Department of Information Technology and
Electrical Engineering, ETH Z\"urich, 8092 Z\"urich, Switzerland (e-mail: \{asasfi, apadoan, doerfler\}@control.ee.ethz.ch).}
\thanks{I. Markovsky is with the Catalan Institution for Research and Advanced Studies, 08010 Barcelona, Spain, and also with the International Centre for Numerical Methods in Engineering, 08034 Barcelona, Spain (e-mail: ivan.markovsky@cimne.upc.edu).}}

\maketitle

\begin{abstract}
This paper introduces an online approach for identifying time-varying subspaces defined by linear dynamical systems. 
The approach of representing linear systems by non-parametric subspace models has received significant interest in the field of data-driven control recently.
This system representation enables us to provide rigorous guarantees for linear time-varying systems, which are difficult to obtain for parametric system models.
The proposed method leverages optimization on the Grassmann manifold leading to the \gls{GREAT}.
We view subspaces as points on the Grassmann manifold and adapt the estimate based on online data by performing optimization on the manifold.
At each time step, a single measurement from the current subspace corrupted by a bounded error is available.
The subspace estimate is updated online using Grassmannian gradient descent on a cost function incorporating a window of the most recent data.
Under suitable assumptions on the signal-to-noise ratio of the online data and the subspace's rate of change, we establish theoretical guarantees for the resulting algorithm.
More specifically, we prove an exponential convergence rate and provide an uncertainty quantification of the estimates in terms of an upper bound on their distance to the true subspace.
The applicability of the proposed algorithm is demonstrated by means of numerical examples.
\end{abstract}


\section{Introduction}
Subspace representations of linear dynamical systems have recently attracted considerable attention in the control and system identification communities~\cite{coulson2019data,coulson2021distributionally,berberich2020data,breschi2023data,willems2005note,verhoek2021data,vanwaarde2020informativity,depersis2019formulas}.
Rooted in behavioral systems theory~\cite{willems1997introduction}, these representations offer a non-parametric framework that describes linear systems as subspaces containing all finite-horizon trajectories.
Unlike traditional parametric models, this approach enables the identification of \gls{LTI} behaviors directly from data that is sufficiently informative~\cite{willems2005note,markovsky2023identifiability}.
The rise of this framework has spurred the development of subspace system identification methods~\cite{van1996subspace} and novel data-driven control techniques~\cite{coulson2019data,coulson2021distributionally,berberich2020data,breschi2023data,willems2005note,verhoek2021data,vanwaarde2020informativity,depersis2019formulas}. 
These approaches exploit the subspace representation of a system, typically defined by the span of a data matrix. 
Empirical studies have shown that the methods exploiting subspace representations of systems often perform remarkably well in various application domains (see, e.g.,~\cite[Sec. 5.2.4.]{markovsky2021behavioral} for an overview).

However, the applicability of these identification and control methods is limited to time-invariant systems.
Recent surveys~\cite{markovsky2021behavioral,berberich2024overview} emphasize this research gap and highlight the challenge of online adaptation.
Another open research challenge is the extension of data-driven control methods to nonlinear systems.
Recent data-driven control methods exploit the fact that nonlinear systems can be over-approximated as \gls{LTV} systems under suitable assumptions on the time-variance~\cite{verhoek2023direct,berberich2022linear}.
We adopt the perspective that trajectories of an \gls{LTV} system always lie in a subspace that may change over time.
Online identification of this non-parametric representation reduces to tracking this subspace.
In this work, we address these challenges by proposing an algorithm that is capable of identifying the subspace representation of \gls{LTV} systems online, with theoretical guarantees on the convergence rate and bias of the estimates.

Subspace identification methods have been extended in the literature to adapt the system model online by recursive updates.
These methods first estimate the subspace spanned by the observability matrix and then compute a state-space representation of the system~\cite{van1996subspace}.
Adaptation is achieved, e.g., by recursively computing the underlying singular value decomposition~\cite{verhaegen1991fast}, or by relying on techniques from the signal processing literature, such as direction of arrival estimation~\cite{lovera2000recursive,oku2002recursive} and the propagator method~\cite{mercere2007convergence,mercere2008propagator}.
These methods primarily address the problem of identifying \gls{LTI} systems recursively. 
To address time-variation of the system, a forgetting factor discounting older data is introduced~\cite{lovera2000recursive,verhaegen1991fast}, possibly in combination with a moving data window~\cite{mercere2008propagator,mercere2007convergence}.
Some of the methods are demonstrated on \gls{LTV} case studies~\cite{lovera2000recursive,mercere2008propagator}, and the derivations in~\cite{verhaegen1991fast} hold approximately for slowly varying systems.
However, rigorous guarantees are not provided for these algorithms when the system is time-varying.
In contrast, we estimate a non-parametric system model characterizing the set of input-output trajectories instead of a particular state-space representation.
This approach enables us to characterize the time-variation and uncertainty of the model through distances between subspaces, allowing for strong theoretical guarantees even for \gls{LTV} systems.

Subspace tracking methods are used extensively in the field of signal processing~\cite{delmas2010subspace,dung2021robust,balzano2018survey,bunch1978updating,yang1995PAST,he2012incremental,balzano2010online,xu2013gosus}, with a wide range of applications~\cite{dung2021robust,balzano2018survey}.
Various frameworks have been developed based on, e.g., incremental singular value decomposition~\cite{bunch1978updating}, recursive least squares method~\cite{yang1995PAST} or the Grassmannian optimization~\cite{he2012incremental,balzano2010online,xu2013gosus,zimmermann2018geometric}.
For some of these algorithms, theoretical guarantees have been derived, see~\cite{vaswani2018robust} for a detailed summary. 
Most of them are asymptotic, but finite-sample guarantees also exist in the literature~\cite{vaswani2017finite,zhang2016global}.
Given the literature's focus on signal processing, an important objective of the available subspace tracking methods is compressing large amounts of data.
Therefore, the dimension of the embedding space is typically large, while the subspace is lower dimensional.
Most of the available algorithms are tailored to reduce the computation complexity by exploiting this property.
However, we focus on subspace representations of deterministic linear dynamical systems~\cite{willems2005note}. 
For these systems, both the embedding space and the subspace are typically of moderate and comparable dimension.
Moreover, existing methods rely on assumptions of independence among samples, which do not hold in the case of dynamical systems where time series are inherently correlated.
Thus, these methods are inherently unsuitable for online system identification.
 
To address these issues, we adopt a geometric approach and view subspaces as points on the Grassmann manifold.
Tracking can be posed as an unconstrained optimization problem on the manifold~\cite{edelman1998geometry,boumal2023introduction,sepulchre2008optimization}.
Techniques to solve unconstrained problems, such as gradient descent, are well-understood~\cite{boyd2004convex,boumal2023introduction} and allow for a simple analysis.
This fact can be exploited to provide deterministic guarantees even in case the data are correlated, generated by linear time-varying dynamical systems, and subject to noise.
Contrary to parametric system identification methods, we optimize over non-parametric subspace models that uniquely define a dynamical system.
These models are coordinate-free, as they are independent of a basis.
By optimizing over subspaces directly in a coordinate-free manner, we obtain an unconstrained problem, avoiding non-convex orthogonality constraints that arise when explicitly optimizing over bases.

The contributions of our work are the following. 
We introduce the \acrfull{GREAT} algorithm for subspace tracking that is suitable for online identification of \gls{LTV} systems from data corrupted by bounded measurement error.
The method minimizes the projection error of an online data window onto the estimated subspace through gradient descent on the Grassmann manifold.
Assuming that the subspace's rate of change is bounded, we provide finite-sample convergence certificates for the algorithm, in the form of guaranteed convergence rate and bias.
The derived bound explicitly depends on the subspace's rate of change, the bound on the measurement error, and the persistency of excitation property of the online data.
In the special case of \gls{LTI} systems and exact measurements, the estimates converge to the true subspace exponentially fast.

The applicability of the \gls{GREAT} algorithm is demonstrated on a numerical example, in which an \gls{LTV} airplane model is identified from noisy data. 
We illustrate the predictive performance of the model obtained by our algorithm compared to two popular system identification methods.
First, we show that the \gls{GREAT} algorithm performs similarly to the well-known subspace identification method N4SID~\cite{van1994n4sid} recomputed at each time step to allow for adaptation.
Second, we benchmark it against a standard online identification method, that estimates a parametric model via the recursive least squares technique~\cite{ljung1999system}.
While the two methods achieve similar prediction error nominally, the proposed gradient-based method exhibits better robustness properties against large measurement errors.
Furthermore, we also compare our method to existing subspace tracking algorithms from the signal processing literature~\cite{balzano2010online,yang1995PAST}.
Finally, we study the conservatism of the theoretical bounds on a synthetic example.

The remainder of the paper is organized as follows. In Section~\ref{sec:prelim}, we introduce the notation and provide preliminaries on subspaces and manifold optimization. 
We formulate the problem and motivate it through online system identification in Section~\ref{sec:setup}.
The \gls{GREAT} algorithm and the corresponding theoretical analysis are presented in Section~\ref{sec:main}.
Section~\ref{sec:num-ex} contains the numerical examples, and Section~\ref{sec:conclusion} concludes the paper.
Most proofs are deferred to the Appendix.

\section{Preliminaries} \label{sec:prelim}
In this section, we introduce notation and differential-geometric concepts that are used later in the paper.

\subsection{Notation}
Let $\mathbb{N}$ ($\mathbb{R})$ denote the set of positive integers (real numbers), and $\mathbb{Z}_{\geq T}$ ($\mathbb{R}_{\geq T}$) is the set of integers (real numbers) greater than or equal to $T$.
We denote the trace of a square matrix $A$ by $\mathrm{tr}(A)$.
Given two matrices $A$ and $B$ of the same dimension, $\langle A,B\rangle_F = \mathrm{tr}(A^\top B)$ denotes the Frobenius inner product, and $\|A\|_F$ is the Frobenius norm of $A$. The $i$-th largest singular value of a matrix $A$ is denoted by $\sigma_i(A)$. 
The identity matrix of size $n\times n$ is $I_n$, and $\mathrm{diag}(d_1,d_2,\dots,d_n)$ denotes a diagonal matrix, with elements $d_1,d_2,\dots,d_n$ on the main diagonal. 
The $\ell_\infty$ norm of a sequence $\{f_t\}_{t\geq T}$ is defined as $\|f\|_\infty = \sup_{t\geq T} |f(t)|$. 
By $\mathcal{K}_\infty$ we denote the class of functions $g:\mathbb{R}_{\geq0} \to \mathbb{R}_{\geq0}$ that are continuous, strictly increasing, and satisfy $g(0) = 0$ and $\lim_{r\to\infty} g(r) = \infty$.

\subsection{Subspaces}
Throughout the paper, we denote subspaces by bold capital letters.
Any orthonormal matrix whose columns span the subspace is denoted by the same capital letter in regular font.
Consider a subspace $\mathbf{U}$ of dimension $d$ in $\mathbb{R}^n$ and one of its matrix representations $U\in\mathbb{R}^{n\times d}$ with $U^\top U = I_d$.
The orthogonal projection onto $\mathbf{U}$ and its orthogonal complement are defined as $P_{\mathbf{U}} = UU^\top$ and $P_{\mathbf{U}}^\perp = I_n - P_{\mathbf{U}}$, respectively.
Note that the projections are independent of the subspace representation $U$. 
Note also that $P_{\mathbf{U}} = P_{\mathbf{U}}^\top = P_{\mathbf{U}}^2$ and $\|P_{\mathbf{U}} M\| \leq \|M\|$ for any matrix $M$ and any matrix norm $\|\cdot\|$~\cite{golub2013matrix,horn2012matrix}.

Now consider another subspace $\mathbf{V}$ of the same dimension as $\mathbf{U}$, with representation $V\in\mathbb{R}^{n\times d},~V^\top V = I_d$. 
The principal angles between $\mathbf{U}$ and $\mathbf{V}$ are denoted by $0\leq \theta_1\leq \dots \leq \theta_d \leq \pi/2$, and can be computed using, e.g., the \gls{SVD}, see~\cite[Sec. 6.4.3]{golub2013matrix} for details.
The following metrics will be used to quantify the distance between two $d$-dimensional subspaces~\cite{ye2016schubert,padoan2022behavioral,edelman1998geometry,sepulchre2008optimization}
\begin{align*}
    & \text{Chordal:} \; d_{2}(\mathbf{U},\mathbf{V}) = \left ( \sum_{i=1}^d \sin^2 \theta_i \right)^{1/2} = \left(\mathrm{tr}(P_{\mathbf{U}}^\perp P_{\mathbf{V}})\right)^{1/2}, \\
    & \text{Gap:} \; d_\infty(\mathbf{U},\mathbf{V}) = \sin \theta_d = \|P_{\mathbf{U}} - P_{\mathbf{V}}\|_2.
\end{align*}
The inequality $d_\infty(\mathbf{U},\mathbf{V}) \leq d_{2}(\mathbf{U},\mathbf{V}) \leq \sqrt d \cdot d_{\infty}(\mathbf{U},\mathbf{V})$ follows from the definition of the metrics immediately. Furthermore, the following property of the gap metric is useful.
\begin{lemma} \label{lemma:gap_metric}
    For any subspaces $\mathbf{U},\mathbf{V}$ of dimension $d$ in $\mathbb{R}^n$ and any $u \in \mathbf{U}$, the following relation holds
    \begin{align*}
        \|P_{\mathbf{V}}^\perp u\|_2 \leq d_\infty(\mathbf{U},\mathbf{V})\cdot\|u\|_2.
    \end{align*}
\end{lemma}
\begin{proof}
    \begin{align*}
        \|(I_n-P_{\mathbf{V}}) u\|_2 = \|(P_{\mathbf{U}}-P_{\mathbf{V}}) u\|_2 \leq \|P_{\mathbf{U}}-P_{\mathbf{V}}\|_2 ~\|u\|_2.
    \end{align*}
\end{proof}

\subsection{Grassmannian geometry}
This section recalls some concepts related to the \textit{Grassmann manifold} that are used throughout the paper.
Formal definitions are outside the scope of this paper, and the interested reader is referred to~\cite{boumal2023introduction} for a thorough exposition.
Informally, the Grassmann manifold is the set of all subspaces in $\mathbb{R}^n$ of a given dimension $d$:
\begin{align*}
    \Grass{n}{d} = \left \{\text{subspaces of dimension $d$ in }\mathbb{R}^n\right \}.
\end{align*}
The tangent space associated to each point on the manifold $\subs{}\in\Grass{n}{d}$ is denoted  by $T_{\subs{}}\Grass{n}{d}$ (see~\cite[Def. 8.33]{boumal2023introduction} for a formal definition).
The tangent bundle of a manifold is the disjoint union of its tangent spaces, endowed with a smooth manifold structure, defined as
\begin{align*}
    T\Grass{n}{d} = \left \{(\subs{},\mathbf{V})~|~\subs{}\in\Grass{n}{d}, \mathbf{V} \in T_{\subs{}}\Grass{n}{d} \right \}.
\end{align*}

Points on the Grassmann manifold are abstract objects.
For computation purposes, we represent a point on $\Grass{n}{d}$ by a matrix in $\mathbb{R}^{n\times d}$, whose columns are orthonormal and span the corresponding subspace\footnote{In fact, the set $\left \{U\in\mathbb{R}^{n\times d}~|~U^\top U = I_d\right \}$ endowed with a differentiable structure is called the Stiefel manifold, and it is an embedded submanifold of the linear space $\mathbb{R}^{n\times d}$. The Grassmannian can be viewed as a quotient manifold of the Stiefel manifold. More details about quotient manifolds and the Grassmannian can be found in~\cite[Chp. 9]{boumal2023introduction}.}.
Furthermore, for any representation $U\in\mathbb{R}^{n\times d},~U^\top U = I_d$ of a point $\subs{} \in \Grass{n}{d}$, all tangent vectors $\mathbf{V} \in T_{\subs{}}\Grass{n}{d}$ admit a representation $V\in\mathbb{R}^{n\times d}$ that satisfies $\subsST{}^\top V = 0$~\cite{boumal2023introduction}.

The Grassmann manifold can be equipped with a Riemannian metric using the Frobenius inner product making it a Riemannian manifold\footnote{More precisely, the Frobenius inner product defines a Riemannian metric on the Stiefel manifold, which induces a Riemannian metric on the Grassmann manifold~\cite[Sec.~9.7]{boumal2023introduction}.}.
Given a differentiable function $f$ that maps from $\Grass{n}{d}$ to $\mathbb{R}$, the Riemannian gradient of $f$ denoted by $\mathrm{grad}~ f$ is a vector field on the manifold. 
The gradient associates a point $\subs{} \in \Grass{n}{d}$ with a tangent vector $\mathbf{V}\in T_{\subs{}} \Grass{n}{d}$, i.e., $\mathbf{V} = \mathrm{grad}~f(\subs{})$.
A representation of $\mathbf{V}$ can be calculated as~\cite[Sec.~9.16]{boumal2023introduction}
\begin{align} \label{eq:riemannian_grad}
    V = P_{\subs{}}^\perp \nabla \bar{f}(\subsST{}),
\end{align}
where $\bar{f}:\mathbb{R}^{n\times d}\to\mathbb{R}$ is such that $\bar{f}(\subsST{}) = f(\subs{})$ for any orthonormal matrix representation $\subsST{}$ spanning the subspace $\subs{}$, and $\nabla \bar{f}(\subsST{})$ is the Euclidean gradient of $\bar{f}$ evaluated at $\subsST{}$.

In manifold optimization, various maps (called retractions) are used to move along the manifold in tangent directions. 
In this work, we use the \textit{exponential map} due to its desirable theoretical properties. 
The exponential map $\mathrm{Exp}_{\mathbf{U}}(\mathbf{V})$ maps an element of the tangent bundle $(\subs{},\mathbf{V}) \in T\Grass{n}{d}$ to another element of the manifold $\subs{}^+ \in \Grass{n}{d}$, i.e., $\subs{}^+ = \mathrm{Exp}_{\mathbf{U}}(\mathbf{V})$.
Consider a representation $(U,V)$ of $(\subs{},\mathbf{V})$,
and let $V = Q_1 S Q_2^\top$ denote the compact \gls{SVD}, where $Q_1\in\mathbb{R}^{n\times d}$, $Q_2\in \mathbb{R}^{d\times d}$, and $S = \mathrm{diag}(\sigma_1,\sigma_2,\dots,\sigma_d)\in \mathbb{R}^{d\times d}$ with $\sigma_i$ denoting the $i$-th singular value of $V$.
Note that $Q_1\in\mathbb{R}^{n\times d}$ only contains columns corresponding to the non-zero singular values.
Then, a representation of $\subs{}^+$ can be computed using the formula~\cite[Eq. (2.65)]{edelman1998geometry}
\begin{align*}
    \subsST{}^+ = [U Q_2~Q_1] 
    \begin{bmatrix}
    \mathrm{diag}(\cos(\sigma_1),\cos(\sigma_2),\dots,\cos(\sigma_d)) \\ \mathrm{diag}(\sin(\sigma_1),\cos(\sigma_2),\dots,\cos(\sigma_d))
    \end{bmatrix}
    Q_2^\top.
\end{align*}

In the remainder of the paper, with a slight abuse of notation, we provide formulas for points in $\Grass{n}{d}$ and $T_{\subs{}}\Grass{n}{d}$ using an arbitrary orthonormal matrix representation of them.
However, this does not undermine consistency, since the quantities we work with, such as distances or projections, are coordinate-free notions defined for subspaces, not their representations.
Consequently, these quantities are invariant under the choice of representation.

\section{Problem setup} \label{sec:setup}
We now describe the problem setup, whose relevance is illustrated later on through the example of \gls{LTV} system identification.

\subsection{Problem formulation} \label{sec:problem}
We consider the problem of recursively estimating an unknown and time-varying subspace $\subs{t} \in \Grass{n}{d}$ based on possibly noisy data, collected online. 
We assume that a sample $\traj{}{t} \in \mathbb{R}^n$ is available at each time, consisting of a nominal part $\ntraj{}{t}$ from $\subs{t}$ corrupted by some measurement error $e_t$. The measurement error is assumed to be bounded, as formalized below.
\begin{assum} \label{ass:meas_err}
    Each sample $\traj{}{t} \in \mathbb{R}^n$ can be decomposed as
    \begin{align*}
        \traj{}{t} = \ntraj{}{t} + e_t,
    \end{align*}
    with $\ntraj{t-L+1}{t}\in\subs{t}$ and $\|e_t\|_2 \leq \epsilon$ for all $t\in\mathbb{N}$.
\end{assum}
Note that we do not make any (probabilistic) assumption on the distribution of the noise $e_t$, which is often unknown.
Instead, we use the upper bound $\epsilon$ from Assumption~\ref{ass:meas_err} to provide worst-case (deterministic) guarantees for the estimate obtained by the algorithm proposed in the subsequent section.

In addition, we assume an upper bound on the subspace's temporal variability.
\begin{assum} \label{ass:LTV-Lipschitz}
    There exists some constant $c\geq0$ such that $d_{2}(\subs{t},\subs{t+1}) \leq c$ holds for all $t \in \mathbb{N}$.
\end{assum}
In case $c=0$, the subspace is time-invariant, i.e., $U_{t+1} = U_{t}$ for all $t\in\mathbb{N}$.

Our estimate of the subspace at time $t\in\mathbb{N}$ is denoted by $\estim{t} \in \Grass{n}{d}$.
The uncertainty of this estimate is quantified by an upper bound on its distance to the true subspace $\subs{t}$.
Therefore, uncertainty sets are in the form of metric balls of radius $r$ centered around some subspace $\subs{}\in\Grass{n}{d}$, defined as
\begin{align*}
    \mathbb{B}_r(\subs{}) := \left \{\estim{}\in\Grass{n}{d}~|~d_{2}(\subs{},\estim{}) \leq r \right \}.
\end{align*}
The problem of interest can now be formulated as follows.
\begin{problem} \label{problem}
    Given an initial estimate $\estim{t_0}$ at time $t_0\in\mathbb{N}$ with $\subs{t_0} \in \mathbb{B}_{r_0}(\estim{t_0})$ for some $r_0$ and data $u_t$, $t\in\mathbb{N}$, find an iterative algorithm that
    \begin{itemize}
        \item provides \textbf{uncertainty quantification} in terms of an invariant tube, that is $\subs{t} \in \mathbb{B}_{r_t}(\estim{t})$ for all $t > t_0$,
        \item yields \textbf{diminishing bias}, that is, $\lim_{t\to\infty} r_t = 0$ if $c=0$ and $\epsilon =0$, and
        \item is \textbf{recursive}, that is, $\estim{t+1}$ depends only on $\estim{t}$ and the past length-$T$ data window $u_\tau,~\tau\in\{t-\windowlength+1,\dots,\\t-1,t\}$ for some $1\leq \windowlength \leq t_0 + 1$.
    \end{itemize}
\end{problem}
The uncertainty set is called an \textit{invariant tube}, as the sequence of metric balls $\{\mathbb{B}_{r_t}(\estim{t})\}_{t>t_0}$ defines a tubular neighborhood around the estimates, in which the true subspaces are guaranteed to remain for all times $t>t_0$.
In this work, we focus on methods that adapt the estimate based on a window of most recent data to simplify the analysis. 
However, achieving adaptation by introducing a forgetting factor that discounts the past data is also possible in the proposed framework, see Remark~\ref{rem:forgetting_factor} in Section~\ref{sec:tracking}.

\subsection{Motivating application: online system identification} \label{sec:mot_ex}
We illustrate the relevance of our problem setup through the example of identifying \gls{LTV} systems in the framework of behavioral systems theory~\cite{willems1997introduction}.
Consider first a state-space representation of a time-varying linear system describing the relationship between inputs and outputs as
\begin{align} 
\begin{split} \label{eq:TV-SS}
    x_{t+1} & = A_t x_t + B_t v_t, \\
    y_t & = C_t x_t + D_t v_t,
\end{split}
\end{align}
with state $x_t \in \mathbb{R}^\sysorder$ and time-dependent matrices $A_t\in \mathbb{R}^{\sysorder \times \sysorder},~B_t \in \mathbb{R}^{\sysorder\times m},~C_t\in \mathbb{R}^{p\times \sysorder}$ and $D_t \in \mathbb{R}^{p\times m}$. 
The input and output trajectories on the time interval $[t,t+L]$ are defined as $v_{[t,t+L]} = [v_t^\top~\dots~v_{t+L}^\top]^\top$ and $y_{[t,t+L]} = [y_t^\top~\dots~y_{t+L}^\top]^\top$, and they can be expressed using~\eqref{eq:TV-SS} in the form
\begin{align} \label{eq:state-space_subspace}
    \begin{bmatrix}
        v_{[t,t+L]} \\ y_{[t,t+L]}
    \end{bmatrix} = \underbrace{\begin{bmatrix}
        0 & I_{m(L+1)} \\ \mathcal{O}_{[t,t+L]} & \mathcal{T}_{[t,t+L]}
    \end{bmatrix}}_{=:\Lambda_t}
    \begin{bmatrix}
        x_t \\ v_{[t,t+L]}
    \end{bmatrix},
\end{align}
where the matrices $\mathcal{O}_{[t,t+L]}\in\mathbb{R}^{p(L+1) \times k}$ and $\mathcal{T}_{[t,t+L]} \in \mathbb{R}^{p(L+1)\times m(L+1)}$ are defined as
\footnotesize
\begin{align*}
  & \mathcal{O}_{[t,t+L]} = \begin{bmatrix}
    C_t \\
    C_{t+1}A_{t} \\
    C_{t+2}A_{t+1}A_t \\
    \vdots \\
    C_{t+L} A_{t+L-1} \cdots A_t
  \end{bmatrix}, \\
  & \mathcal{T}_{[t,t+L]} = \\
  & \begin{bmatrix}
      D_t & 0 & 0 & \dots & 0 \\
        C_{t+1}B_t & D_{t+1} & 0 & \dots & 0 \\
        C_{t+2}A_{t+1}B_t & C_{t+2}B_{t+1} & D_{t+2} &\dots  & 0\\
        \vdots & \vdots & \vdots & \ddots & \vdots \\
        C_{t+L}A_{t+L-1}\dots A_{t+1}B_t & \dots & \dots & \dots & D_{t+L}   
  \end{bmatrix}.
\end{align*}
\normalsize
The matrix $\mathcal{O}_{[t,t+L]}$ is related to the observability of the system~\cite{rugh1996linear}. 
If it has full column rank, the matrix $\Lambda_t$ has full rank, i.e., it spans a subspace of dimension $\sysorder+m(L+1)$ for all $t\in \mathbb{N}$.
Techniques for estimating the system order~$\sysorder$, and hence, the dimension of the subspace representation, are available for \gls{LTI} systems in the subspace identification literature~\cite{van1996subspace}.

Note that existing works on (recursive) subspace identification~\cite{van1996subspace,verhaegen1991fast,oku2002recursive,lovera2000recursive,mercere2007convergence,mercere2008propagator} estimate the subspace spanned by $\mathcal{O}_{[t,t+L]}$ first, and then compute the state-space matrices $A_t,B_t,C_t$ and $D_t$.
Instead of this approach, we focus on identifying the set of all input-output trajectories on the interval $[t,t+L]$, which is called the \textit{restricted behavior}~\cite{markovsky2023identifiability}.
The theoretical foundations of this system description is the subject of behavioral systems theory~\cite{willems1997introduction}.
The restricted behavior of the \gls{LTV} system represented by~\eqref{eq:TV-SS} is defined as
\begin{align*}
    & \mathcal{B}_{[t,t+L]} = \left \{
    \begin{bmatrix} 
        v_{[t,t+L]} \\
        y_{[t,t+L]}
    \end{bmatrix}
    \in \mathbb{R}^{(p+m)\cdot(L+1)} ~\Big | \right . \\
    & \quad \left. \vphantom{\begin{bmatrix} 
        v_{[t,t+L]} \\
        y_{[t,t+L]}
    \end{bmatrix}}
    \exists x_t \text{ such that } x_{[t,t+L]},v_{[t,t+L]}, y_{[t,t+L]} \text{ satisfy~\eqref{eq:TV-SS}} \right \}.
\end{align*}
Due to the linearity of equation~\eqref{eq:state-space_subspace}, the restricted behavior $\mathcal{B}_{[t,t+L]}$ is a subspace spanned by the columns of $\Lambda_t$.
For \gls{LTI} systems, this subspace is time-invariant, i.e., $\mathcal{B}_{[t,t+L]} \equiv \mathcal{B}$ for all $t\in\mathbb{N}$, and input-output data can be used directly to form a basis of $\mathcal{B}$.
Suppose that input and output measurements at $t=\{1,2,\dots,T_d\}$ are available offline with $T_d\gg L$.
Let us construct the Hankel matrix of the inputs (and similarly of the outputs) as
\begin{equation*}
   \mathcal{H}_{L+1}(v) = \begin{bmatrix}
       v_1 & v_2 & \dots & v_{T_d-L} \\
       v_2 & v_3 & \dots & v_{T_d-L+1} \\
       \vdots & \vdots & \ddots & \vdots \\
       v_{L+1} & v_{L+2} & \dots & v_{T_d}
   \end{bmatrix}.
\end{equation*}
In case the input is persistently exciting, the block Hankel matrix $\begin{bmatrix}
    \mathcal{H}_{L+1}(v) \\ \mathcal{H}_{L+1}(y)
\end{bmatrix}$ spans the behavior $\mathcal{B}$~\cite{willems2005note}.
However, the system may be time-varying in our setup, and therefore, the restricted behavior does not attain a direct data representation.

In light of the above, online identification of \gls{LTV} systems in the behavioral setting is equivalent to tracking the time-varying subspace $\mathcal{B}_{[t,t+L]}$, and thus, it can be formulated as Problem~\ref{problem}.
At each time $t$, a sample $\traj{}{t} = [v_{[t-L,t]}^\top~y_{[t-L,t]}^\top]^\top$ from the subspace $\subs{t} = \mathcal{B}_{[t-L,t]}$ can be constructed from the most recent input-output measurements.
Note that $\mathcal{B}_{[t-L,t]}$ is a special time-varying subspace that represents the restricted behavior of an LTV system.
However, we state the main results of our work in Section~\ref{sec:main} for general time-varying subspaces.

The proposed problem formulation allows us to handle additive measurement error on the inputs and outputs, c.f., Assumption~\ref{ass:meas_err}.
This error-in-variables formulation is used in many domains, such as identifying electrical circuits~\cite{brouillon} or financial models~\cite{MADDALA1996507}.
Considering this setup is also natural in the behavioral framework, since inputs and outputs are not distinguished in this setting~\cite{willems1997introduction}.
However, the input is often known exactly in the context of control for example.
The upper bound on the distance between two consecutive subspaces in Assumption~\ref{ass:LTV-Lipschitz} naturally quantifies how fast the behavior changes.
On the other hand, characterizing the time-variation of the system by the change in $A_t,B_t,C_t$ is problematic, as these parameters are basis dependent.
Nevertheless, if the matrices $A_t,B_t,C_t$ and $D_t$ vary slowly with time, the change in the matrix on the right-hand side of equation~\eqref{eq:state-space_subspace} is small, and hence, $d_2(\mathcal{B}_{[t,t+L]},\mathcal{B}_{[t+1,t+L+1]})$ can be bounded (c.f.,~\cite{wedin1972perturbation}).
The uncertainty characterization of the identified behavior is useful for various downstream tasks, such as prediction, estimation, or control.
For example, it can be used to quantify the uncertainty of predicted trajectories, or to robustify control formulations.
Besides describing time-varying systems, \gls{LTV} models are also often used to approximate the behavior of nonlinear systems~\cite{verhoek2023direct,berberich2022linear}.

Even though the state-space representation~\eqref{eq:TV-SS} can be related to the restricted behavior through the equation~\eqref{eq:state-space_subspace}, we do not aim at recovering the matrices $A_t,B_t,C_t$ and $D_t$ to form an estimate of $\mathcal{B}_{[t,t+L]}$.
Instead, we use the subspace estimate and its uncertainty quantification directly to perform downstream tasks, such as prediction, which is described below and illustrated in Section~\ref{sec:ex2}.
We divide the input-output trajectories into initial and future parts of length $T_{\mathrm{ini}}$ and $T_{\mathrm{fut}}$, respectively, i.e., $v_{[t-T_{\mathrm{ini}}+1,t+T_{\mathrm{fut}}]}=[v_\mathrm{ini}^\top~v_\mathrm{fut}^\top]^\top$ and $y_{[t-T_{\mathrm{ini}}+1,t+T_{\mathrm{fut}}]}=[y_\mathrm{ini}^\top~y_\mathrm{fut}^\top]^\top$.
Let $\estimST{}$ be a basis for the estimated behavior $\mathcal{B}_{[t-T_{\mathrm{ini}}+1,t+T_{\mathrm{fut}}]}$.
We denote the block-rows of $\estimST{}$ corresponding to each trajectory component by $\estimST{}^{v_\mathrm{ini}},\estimST{}^{y_\mathrm{ini}},\estimST{}^{v_\mathrm{fut}}$, and $\estimST{}^{y_\mathrm{fut}}$.
Given $v_\mathrm{ini}, y_\mathrm{ini}, v_\mathrm{fut}$, the future outputs at times $t+1,\dots,t+ T_\mathrm{fut}$ are predicted as $\hat{y}_\mathrm{fut} = M [v_\mathrm{ini}^\top~y_\mathrm{ini}^\top~ v_\mathrm{fut}^\top]^\top$, with $M$ defined as
\begin{align*}
    M := \estimST{}^{y_\mathrm{fut}} 
    \begin{bmatrix}
        \estimST{}^{v_\mathrm{ini}} \\ \estimST{}^{y_\mathrm{ini}} \\
        \estimST{}^{v_\mathrm{fut}}
    \end{bmatrix}^\dagger,
\end{align*}
where $\dagger$ denotes the Moore-Penrose pseudoinverse.
This linear multi-step predictor originates from the classical and widely adopted Subspace Predictive Control~\cite{favoreel1999spc,markovsky2021behavioral}.

\section{Subspace tracking with guarantees} \label{sec:main}
Our main results are presented in this section. 
In Section~\ref{sec:SNR}, we discuss how the signal in the data can be distinguished from noise.
Then, we introduce the gradient descent method in Section~\ref{sec:GGD} that serves as the main building block of the subspace tracking algorithm proposed in Section~\ref{sec:tracking}.

\subsection{Distinguishing between signal and noise} \label{sec:SNR}
In this section, we define and analyze the signal and noise components of the online data.
We consider a data window of length $T\geq d$ consisting of the most recent samples up to time $t$ organized into a matrix
\begin{align*}
    W_t = [\traj{}{t-\windowlength+1}~ \traj{}{t-\windowlength+2}~ \dots ~\traj{}{t}] \in \mathbb{R}^{n \times \windowlength}, \quad t\in\mathbb{Z}_{\geq \windowlength}.
\end{align*}
Our goal is to estimate the current subspace $\subs{t}$ using the data matrix $W_t$.
Note that in case the samples $\traj{}{t}$ are generated by an \gls{LTI} system, $W_t$ is a block Hankel matrix and its columns span the restricted behavior $\mathcal{B}_{[t,t+L]}$ (see Section~\ref{sec:mot_ex} and Remark~\ref{rem:Fundamental_lemma} for details).
However, when the system is time-varying, only the last sample $u_t$ originates from $\subs{t}$.
Therefore, instead of using the nominal sample $\bar{u}_t$ and the measurement error $e_t$ from Assumption~\ref{ass:meas_err}, we characterize the signal and noise components of $W_t$ by introducing the following decomposition:
\begin{align*}
    W_t = P_{\subs{t}} W_t + P_{\subs{t}}^\perp W_t.
\end{align*}
The first term lies in $\subs{t}$ and is defined to be the signal, while the second term lies in the orthogonal complement of $\subs{t}$ and is considered as noise (see Figure~\ref{fig:SNR} for an illustration).
Under the stated assumptions, the noise term can be bounded as follows. The proof can be found in the Appendix.
\begin{lemma}[Noise bound] \label{lemma:LTV-Hankel-noisy}
    Let Assumptions~\ref{ass:meas_err} and \ref{ass:LTV-Lipschitz} hold. Then, 
    \begin{align}
    \begin{split} \label{eq:delta_def}
            \|P_{\subs{t}}^\perp {W}_t\|_F \leq & \delta_t, \quad \forall t \geq T,
            \end{split}
    \end{align}
    where $\delta_t \geq 0$ is strictly monotone in $\epsilon$ and $c$, defined as
    \begin{align} \label{eq:delta_def_true}
        \delta_t:= c~ \|W_tD\|_F + \epsilon \sqrt{\windowlength} \left( c (\windowlength-1) + 1 \right),
    \end{align}
    with $D:=\mathrm{diag}(\windowlength-1,\windowlength-2,\dots,0)$.
\end{lemma}

Lemma~\ref{lemma:LTV-Hankel-noisy} guarantees that the noise component $P_{\subs{t}}^\perp W_t$, which captures both the (slowly) time-varying nature of the subspace and the effect of measurement error, is bounded by the function $\delta_t$.
To be able to guarantee positive convergence rate for the updates in the subsequent section, we also impose bounds on the singular values of the signal component $P_{\subs{t}} W_t$.
\begin{assum} \label{ass:PE}
    There exist constants $\underline{\sigma}$ and $\overline{\sigma}$ where $0 < \underline{\sigma} \leq \overline{\sigma}$ such that the matrix $W_t$ satisfies $\sigma_1(P_{\subs{t}} W_t) \leq \overline{\sigma}$ and $\sigma_d(P_{\subs{t}} W_t) \geq \underline{\sigma}$ for all $t\in\mathbb{Z}_{\geq \windowlength}$.
\end{assum}
Choosing $\overline{\sigma}$ as any uniform upper bound on $\|W_t\|_2$ trivially satisfies the assumption, since $\|W_t\|_2 = \sigma_1(W_t) \geq \sigma_1(P_{\subs{t}} W_t)$.
Furthermore, the lower bound ensures that the signal term $P_{\subs{t}}W_t$ spans $\subs{t}$ as illustrated in Figure~\ref{fig:SNR}.

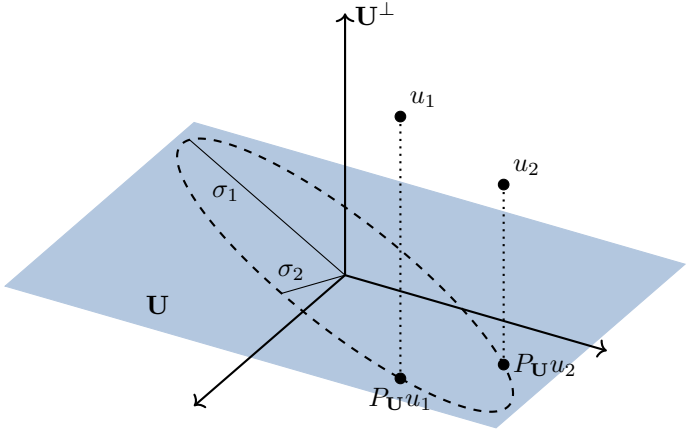
\begin{figure}
    \centering
    \tdplotsetmaincoords{60}{120}

\begin{tikzpicture}[tdplot_main_coords]

\def\plane{3.75}
\def\planex{-2.5}
\def\planey{-3.75}
\filldraw[myBlue!30] (\planex,\planey,0) -- (-\planex,\planey,0) -- (-\planex,\plane,0) -- (\planex,\plane,0) -- cycle;

\draw[thick,->] (0,0,0) -- (4,0,0) node[anchor=north east]{};
\draw[thick,->] (0,0,0) -- (0,4,0) node[anchor=north west]{};
\draw[thick,->] (0,0,0) -- (0,0,4) node[anchor=west]{$\mathbf{U}^\perp$};

\def\pex{2}
\def\pey{2}
\def\pkx{1}
\def\pky{3}

\coordinate (P1) at (\pex, \pey, 4);
\coordinate (P2) at (\pkx, \pky, 2.75);

\coordinate (P1_proj) at (\pex, \pey, 0);
\coordinate (P2_proj) at (\pkx, \pky, 0);

\filldraw[black] (P1) circle (2pt) node[anchor=south west, black]{$u_1$};
\filldraw[black] (P2) circle (2pt) node[anchor=south west, black]{$u_2$};

\filldraw[black] (P1_proj) circle (2pt) node[below, black]{$P_{\mathbf{U}} u_1$};
\filldraw[black] (P2_proj) circle (2pt) node[right, black]{$P_{\mathbf{U}} u_2$};

\draw[dotted,thick] (P1) -- (P1_proj);
\draw[dotted,thick] (P2) -- (P2_proj);

\coordinate (ax1) at (-2.0733, -3.5726, 0);
\coordinate (ax2) at (0.8376, -0.4861, 0);

\draw[thick, black, samples=100, smooth, dashed, domain=0:360] 
    plot ({\pex*cos(\x) + \pkx*sin(\x)}, {\pey*cos(\x) + \pky*sin(\x)}, 0);

\draw[black] (0,0,0) -- (ax1) node[below=20pt,right=5pt]{$\sigma_1$};
\draw[black] (0,0,0) -- (ax2) node[above=8pt,right=-5pt]{$\sigma_2$};

\node[below] at (1.5,-2,0) {$\mathbf{U}$};

\end{tikzpicture}
    \caption{Illustration of the signal-to-noise ratio properties on a two dimensional subspace $\mathbf{U}$ in $\mathbb{R}^3$ (shaded). Two data samples, $u_1$ and $u_2$ are depicted along with their noise components (dotted lines). 
    The signal component is $P_{\subs{}}W$ with $W=[u_1~u_2]$ and singular values $\sigma_1$ and $\sigma_2$.
    The image and the 2-norm condition of the matrix $P_{\subs{}}W$ are illustrated by the dashed ellipse.
    Note that, since $\sigma_2>0$, the matrix $P_{\subs{}}W$ spans the whole subspace $\mathbf{U}$, i.e., the ellipse lies in the plane~$\mathbf{U}$.
    }
    \label{fig:SNR}
\end{figure}

\begin{remark} \label{rem:Fundamental_lemma}
Lemma~\ref{lemma:LTV-Hankel-noisy} and Assumption~\ref{ass:PE} are related to system identification described in Section~\ref{sec:mot_ex} as follows.
Each sample $\traj{}{t}$ is an $L+1$ step long input-output trajectory $[v_{[t-L,t]}^\top~y_{[t-L,t]}^\top]^\top$, and therefore, $W_t = \begin{bmatrix}
    \mathcal{H}_{L+1}(v) \\\mathcal{H}_{L+1}(y)
\end{bmatrix}$.
Suppose that the data is generated by an \gls{LTI} system, implying $\subs{t} = \subs{t+1} = \subs{}$ for all $t\in \mathbb{Z}_{\geq \windowlength}$, and therefore, $c=0$.
Assume further that the input-output data are measured exactly, i.e., $\epsilon=0$.
Moreover, let the system be observable implying that $\subs{}$ is of dimension $\sysorder+m(L+1)$.
Under these assumptions, the bound in Lemma~\ref{lemma:LTV-Hankel-noisy} reduces to $\delta_t \equiv 0$.
Consequently, it holds that $\|P_{\subs{}}^\perp {W}_t\|_F=0$, i.e., all measured trajectories of length $L+1$ must lie in the subspace $\subs{}$ of dimension $\sysorder+m(L+1)$, which is in line with the linearity and time-invariance of the system.
Furthermore, the condition $\sigma_d(P_{\subs{t}} W_t) \geq \underline{\sigma} >0$ from Assumption~\ref{ass:PE} ensures that the singular values of the Hankel matrix $W_t$ are lower bounded by $\underline{\sigma}$.
This can be guaranteed by imposing a \textit{quantitative persistency of excitation} condition on the input sequence, as formalized in~\cite{coulson2022quantitative,berberich2023quantitative}.
\hfill\QED 
\end{remark}

\begin{remark} \label{rem:window_length}
    Both the noise bound $\delta_t$ in Lemma~\ref{lemma:LTV-Hankel-noisy} and the lower bound $\underline{\sigma}$ on the signal in Assumption~\ref{ass:PE} grow as the window length $\windowlength$ increases.
    On the one hand, samples from ``older" subspaces are included in the data matrix $W_t$ for larger $\windowlength$, possibly leading to larger noise contribution.
    On the other hand, including more data also increases the signal in $W_t$.
    In fact, the lower bound $\sigma_d(P_{\subs{t}}W_t) \geq \underline{\sigma}>0$ in Assumption~\ref{ass:PE} can only be satisfied if $\windowlength\geq d$ holds.
    From a practical perspective, the choice of $\windowlength$ provides a tuning parameter to control the rate of adaptation of the tracking algorithm introduced in Section~\ref{sec:tracking}.
    A shorter data window leads to faster adaptation, yet choosing $\windowlength$ too small increases the sensitivity to disturbances.
    One can tune the value of $\windowlength$, e.g., by validation, which is demonstrated in the numerical example in Section~\ref{sec:ex2}.
\hfill \QED
\end{remark}

\subsection{Gradient descent on the Grassmannian} \label{sec:GGD}
The main idea of our approach is to update a subspace estimate each time a new sample is available by performing $K$ iterations of gradient descent on the Grassmann manifold.
Due to the simplicity of the gradient descent update rule, we are able to provide strong theoretical guarantees for the resulting algorithm.
We denote the estimate of the true subspace $\subs{t}$ at iteration $k$ by $\estimGGD{t}{k}$.
Since estimates are points on the Grassmann manifold, the optimization is carried out intrinsically on the manifold as described below.

As discussed in the previous section, the data matrix $W_t$ serves as a surrogate representation to the true (unknown) subspace $\subs{t}$.
Therefore, we use the projection error of $W_t$ onto $\estimGGD{t}{k}$ as the cost for the gradient descent:
\begin{align} \label{eq:GGD_cost}
    F_{W_t}(\estimGGD{t}{k}) := \sum_{i=0}^{\windowlength-1} \|P_{\estimGGD{t}{k}}^\perp \traj{t-L-i+1}{t-i} \|_2^2= \|P_{\estimGGD{t}{k}}^\perp W_t\|_F^2.
\end{align}
It is worth noting that if $W_t = \subsST{t}$, the cost function reduces to the squared chordal distance, i.e., $F_{\subsST{t}}(\estimGGD{t}{k}) = d_{2}(\estimGGD{t}{k},\subs{t})^2$ holds.
Using~\eqref{eq:riemannian_grad}, the Riemannian gradient of the cost is
\begin{align} \label{eq:gradient}
\begin{split}
    \mathrm{grad}~F_{W_t}(\estimGGD{t}{k}) & = P_{\estimGGD{t}{k}}^\perp \nabla F_{W_t}(\estimGGD{t}{k}) \\
    & = P_{\estimGGD{t}{k}}^\perp \nabla ~\mathrm{tr}\left( W_t^\top P_{\estimGGD{t}{k}}^\perp W_t \right) \\
    &= -2P_{\estimGGD{t}{k}}^\perp W_t W_t^\top \estimGGDST{t}{k}.
\end{split}
\end{align}
We use the exponential map to take a step of size $\alpha>0$ towards the negative gradient direction, yielding the update for $\estimGGD{t}{k+1}$ and the formula to compute a representation $\estimGGDST{t}{k+1}$
\begin{align} \label{eq:GGD_update_rule}
\begin{split}
    &\estimGGD{t}{k+1}  =  \mathrm{Exp}_{\estimGGD{t}{k}}(-\alpha~\mathrm{grad}~F_{W_t}(\estimGGD{t}{k})), \\
    &\estimGGDST{t}{k+1} = [\estimGGDST{t}{k} Q_2~Q_1] 
    \begin{bmatrix}
    \mathrm{diag}(\cos(-\alpha \sigma_1),\dots,\cos(-\alpha \sigma_d)) \\ \mathrm{diag}(\sin(-\alpha \sigma_1),\dots,\sin(-\alpha \sigma_d))
    \end{bmatrix}
    Q_2^\top,
\end{split}
\end{align}
where $\mathrm{grad}~F_{W_t}(\estimGGD{t}{k}) = Q_1 \mathrm{diag}(\sigma_1,\sigma_2,\dots,\sigma_d)Q_2^\top$ is the compact \gls{SVD}.
We use a constant step size for the update to streamline the analysis. 
However, the performance of the algorithm can be improved by optimizing the step size online (e.g., via line search), as illustrated in the numerical example in Section~\ref{sec:ex2}.

\begin{remark} \label{rem:GROUSE}
    Another subspace tracking algorithm, termed GROUSE~\cite{balzano2010online}, also relies on gradient descent techniques on the Grassmann manifold.
    In contrast to our formulation, the GROUSE algorithm minimizes the projection error of the most recent sample only, and thus, its cost 
    is a special case of~\eqref{eq:GGD_cost} with $\windowlength = 1$.    
    As a result, the gradient descent step boils down to a rank-one update of the subspace estimate yielding lower computational complexity than that of our algorithm.
    In~\cite{zhang2016global}, the authors analyze the GROUSE algorithm under the crucial assumption that the nominal samples $\ntraj{}{t}$ are i.i.d. random variables with support on $\subs{t}$.
    This condition allows them to prove convergence of the estimates in expectation.
    However, the assumption does not hold if the subspace represents a dynamical system (c.f., Sec.~\ref{sec:mot_ex}), where samples are correlated.  
    In our setup, on the other hand, a moving data window consisting of $\windowlength \geq d$ samples is used in the cost. 
    This enables us to impose Assumption~\ref{ass:PE}, which can be satisfied by subspaces representing dynamical systems.
    Under this assumption, we can guarantee monotonic improvement of the subspace estimates in Lemma~\ref{lemma:one-step-improvement-noisy}.
    
    In~\cite{zimmermann2018geometric}, the GROUSE algorithm was used to perform model reduction for nonlinear dynamical systems online.
    The updated subspace estimate is guaranteed to contain the latest sample from the true subspace.
    In contrast, by generalizing the GROUSE algorithm, we are able to provide an upper bound on the distance between the true and estimated subspaces in Lemma~\ref{lemma:one-step-improvement-noisy}, which is a much stronger result.
\hfill\QED 
\end{remark}

\begin{remark} \label{rem:PAST}
Various subspace tracking algorithms were proposed in~\cite{yang1995PAST} to minimize a cost function similar to~(5).
Contrary to our approach, these methods optimize over non-unique matrix representations of the subspace, which are not required to be orthonormal.
For example, the Projection Approximation Subspace Tracking (PAST) algorithm minimizes an approximation of the cost leading to an unconstrained problem that is quadratic in the decision variable.
Another approach proposed in~\cite{yang1995PAST} is to apply gradient descent on the Euclidean space $\mathbb{R}^{n\times d}$ to minimize the original cost.
An analysis of the PAST algorithm was provided in~\cite{yang1996asymptotic}, yet the guarantees only hold for time-invariant subspaces and are asymptotic. 
In contrast, we optimize the subspace directly by posing the problem on the Grassmann manifold, eliminating the issue of non-uniqueness of matrix representations. 
This approach allows us to derive strong, non-asymptotic guarantees for a possibly time-varying reference subspace, which is the main advantage of our formulation compared to those in~\cite{yang1995PAST,yang1996asymptotic}.
\hfill\QED 
\end{remark}

The following enabling lemma establishes guarantees on the evolution of the distance between the estimate and the true subspace under the update rule~\eqref{eq:GGD_update_rule}.
\begin{lemma}[Single step decay bound] \label{lemma:one-step-improvement-noisy}
    Suppose Assumptions~\ref{ass:meas_err}, \ref{ass:LTV-Lipschitz} and \ref{ass:PE} hold and consider some radius $r>0$. Then, for all $\estimGGD{t}{k}\in \mathbb{B}_{r}(\subs{t})$ and all $k \in [0,K-1]$, $t \in \mathbb{Z}_{\geq \windowlength}$
    \begin{align*}
        d_{2}(\estimGGD{t}{k+1},\subs{t})^2  \leq &  d_{2}(\estimGGD{t}{k},\subs{t})^2 \\
        &- \rho~ \| \mathrm{grad}~ d_{2}(\estimGGD{t}{k},\subs{t})^2\|_F^2 
        + \gamma_r(\delta_t),
    \end{align*}
        where $\rho := \alpha \underline{\sigma}^2 - 2\alpha^2\overline{\sigma}^4$ and $\gamma_r \in \mathcal{K}_\infty$ is
    \begin{align} \label{eq:gamma_def}
    \begin{split}
    \gamma_r(\delta_t) := & 8\alpha r\overline{\sigma}(1+4\alpha\overline{\sigma}^2)\delta_t +
    \left( 4\alpha r + 16 \alpha^2 \overline{\sigma}^2(r+2) \right) \delta_t^2 \\
    & + 32\alpha^2\overline{\sigma}\delta_t^3 + 8\alpha^2 \delta_t^4,
    \end{split}
\end{align}
with $\delta_t\geq0$ defined in~\eqref{eq:delta_def_true}.
\end{lemma}
Lemma~\ref{lemma:one-step-improvement-noisy} provides a bound on the improvement of the estimate during a single gradient descent step.
The improvement in the chordal distance between the estimate and the true subspace is proportional to the squared gradient norm times the constant $\rho$, which is related to the signal component of $W_t$.
On the other hand, the noise component of $W_t$ deteriorates the guaranteed decay bound. 
This effect is characterized by the bias term $\gamma_r(\delta_t)$, where $\gamma_r$ is a fourth-order polynomial in the noise bound $\delta_t$ from Lemma~\ref{lemma:LTV-Hankel-noisy}. 
Since all coefficients are non-negative, and the constant coefficient is zero, $\gamma_r$ is a $\mathcal{K}_\infty$ function.
Therefore, for the noise-free case with $\delta_t\equiv0$ (c.f., Remark~\ref{rem:Fundamental_lemma}), Lemma~\ref{lemma:one-step-improvement-noisy} guarantees monotonic improvement of the estimates as long as the gradient of $d_2(\estimGGD{t}{k},\subs{t})$ is non-zero, and the step size is sufficiently small.
More specifically, if the step size $\alpha$ is in the interval $(0,\underline{\sigma}^2/(2\overline{\sigma}^4))$, $\rho$ is guaranteed to be positive; see Figure~\ref{fig:rho} and the later Remark~\ref{rem:step-size} for a more detailed discussion.
Note that the bound in Lemma~\ref{lemma:one-step-improvement-noisy} is tighter for smaller radius $r$, i.e., when the estimate $\estimGGD{t}{k}$ is closer to $\subs{t}$.

\begin{figure}
    \centering
    \begin{tikzpicture}
\def\so{2}
\def\su{1}
\def\marg{0.005}
\def\margy{0.005}
    \begin{axis}[
        axis equal,
        height=5cm,
        width=\linewidth,
        ytick={\su^4/(8*\so^4)},
        yticklabels={$\dfrac{\underline{\sigma}^4}{8\overline{\sigma}^4}$},
        scaled y ticks = false,
        xtick={0,\su^2/(4*\so^4),\su^2/(2*\so^4)},
        xticklabels={$0$,$\dfrac{\underline{\sigma}^2}{4\overline{\sigma}^4}$,$\dfrac{\underline{\sigma}^2}{2\overline{\sigma}^4}$},
        scaled x ticks = false,
        axis lines = center,
        xlabel = $\alpha$,
        ylabel = $\rho$,
        xmin = -\marg, xmax = \su^2/(2*\so^4)+\marg,
        ymin = -\margy, ymax = \su^4/(8*\so^4)+\margy,
        samples = 100,
        domain = -\marg:\su^2/(2*\so^4)+\marg,
    ]
        \addplot[
            smooth,
            thick,
            myBlue
        ]{x*(\su^2 - x*2*\so^4)};
        \addplot[dashed,samples=2,domain=0:\su^2/(4*\so^4)]{\su^4/(8*\so^4)};

        \addplot[dashed] coordinates {
        (\su^2/\so^4/4, 0)
        (\su^2/4/\so^4, \su^4/8/\so^4)
        };
    \end{axis}
\end{tikzpicture}
    \caption{Illustration of the convergence rate $\rho$ as a function of the step size $\alpha$ in Lemma~\ref{lemma:one-step-improvement-noisy}. The convergence rate is positive for any $\alpha$ from the interval $\left(0,\underline{\sigma}^2/(2\overline{\sigma}^2)\right)$. The maximal value of $\rho$ is $\underline{\sigma}^4/(8\overline{\sigma}^2)$, which is achieved with $\alpha=\underline{\sigma}^2/(4\overline{\sigma}^2)$.}
    \label{fig:rho}
\end{figure}
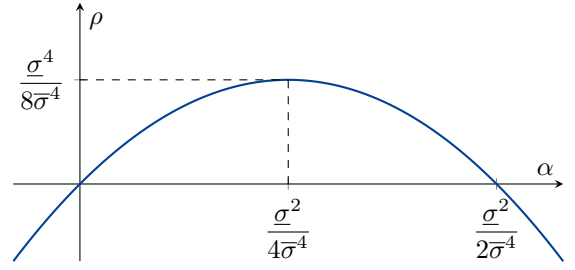

\subsection{Tracking of time-varying subspaces} \label{sec:tracking}
The proposed approach for tracking $\subs{t}$ called \acrfull{GREAT} is described by the pseudo-code in Algorithm~\ref{alg:GGD-exp}.
\begin{algorithm}[H] 
    \caption{\gls{GREAT}}
    \label{alg:GGD-exp}
    \begin{algorithmic}
        \STATE \textbf{Input:} initial estimate $\estim{t_0} \in \Grass{n}{d}$ for some initial time $t_0\geq d$, sequence of samples $\{u_t\}_{t\in\mathbb{N}}$, window length $\windowlength \in \{d,\dots,t_0+1\}$, step size $\alpha\in\left(0,\dfrac{\underline{\sigma}^2}{2\overline{\sigma}^4} \right)$, gradient descent iteration number $K>0$
        \STATE \textbf{for} $t=t_0+1,t_0+2,\dots$
        \STATE \hspace{0.2cm} Construct matrix $W_t = [\traj{}{t-\windowlength+1}~\traj{}{t-\windowlength+2}~\dots~\traj{}{t}]$
        \STATE \hspace{0.2cm} Initialize gradient descent $\estimGGD{t}{0} = \estim{t-1}$
        \STATE \hspace{0.2cm} \textbf{for} $k=0,\dots,K-1$
        \STATE \hspace{0.4cm} Grad. descent step $\estimGGD{t}{k+1} = \mathrm{Exp}_{\estimGGD{t}{k}}(-\alpha~ \mathrm{grad}\, F_{W_t}(\estimGGD{t}{k}))$
        \vspace{-4mm}
        \STATE \hspace{0.2cm} \textbf{end}
        \STATE \hspace{0.2cm} Update estimate $\estim{t} = \estimGGD{t}{K}$
        \STATE \textbf{end}
        \STATE \textbf{Output:} sequence of estimates $\{\estim{t}\}_{t > t_0}$
    \end{algorithmic}
\end{algorithm}
The \gls{GREAT} algorithm recursively updates the subspace estimate $\estim{t}$ by performing $K$ steps of gradient descent update~\eqref{eq:GGD_update_rule} between sampling times $t$ and $t+1$.  
The initial estimate $\estim{t_0}$ can be constructed as the span of $d$ singular vectors of an initializing data matrix $W_\mathrm{ini} = [\traj{}{1}~\traj{}{2}~\dots~\traj{}{t_0}]$, corresponding to its largest singular values.
With this choice, $\estim{t_0}$ becomes the minimizer of the cost $F(\estim{}) = \|P^\perp_{\estim{}} W_\mathrm{ini}\|_F^2$ (c.f., Eq.~\eqref{eq:GGD_cost}).
During online operation, the data matrix $W_t$ is constructed from the most recent samples, and therefore, the subspace estimates adapt to the online data.

\begin{remark} \label{rem:computation}
    The computationally most expensive part of Algorithm~\ref{alg:GGD-exp} is calculating the gradient of the cost in~\eqref{eq:gradient} at each gradient descent iteration.
    For an efficient implementation, one can exploit that the data matrix only appears as an inner product $W_t W_t^\top$ (the empirical data covariance) in the expression.
    This inner product at time $t$ can be formulated as a rank-2 update , i.e., $W_tW_t^\top = W_{t-1}W_{t-1}^\top - \traj{}{t-\windowlength}\traj{}{t-\windowlength}^\top + \traj{}{t}\traj{}{t}^\top$.
    Therefore, the complexity of the algorithm does not depend on $\windowlength$. 
    With this implementation, the complexity of computing the gradient is $\mathcal{O}(n^2d)$, leading to overall complexity $\mathcal{O}(K n^2d)$ per sample.
    The other expensive part of the algorithm is evaluating the exponential map in equation~\eqref{eq:GGD_update_rule}, which requires computing the compact \gls{SVD} of the gradient.
    However, that step is of complexity $\mathcal{O}(n d^2)$ only, and therefore it is not the bottleneck.
    Similarly to Algorithm~\ref{alg:GGD-exp}, various data-driven control formulations depend only on the empirical data covariance, and this fact was exploited for computational reasons, e.g., in~\cite{zhao2024data}.

    Note that there exist computationally more efficient algorithms for subspace tracking in the literature. 
    For example, the GROUSE~\cite{balzano2010online} and \gls{PAST}~\cite{yang1995PAST} algorithms have $\mathcal{O}(nd^2)$ and $\mathcal{O}(nd)$ complexity per iteration, respectively. 
    However, we are able to provide non-asymptotic guarantees for the \gls{GREAT} algorithm, even when the subspace represents a dynamical system (see Remarks \ref{rem:GROUSE} and \ref{rem:PAST}).
\hfill\QED 
\end{remark}

Now we turn to the theoretical analysis of Algorithm~\ref{alg:GGD-exp}.
We provide convergence guarantees for the \gls{GREAT} algorithm by utilizing the bound in Lemma~\ref{lemma:one-step-improvement-noisy} that quantifies the improvement of the inner loop with the gradient descent update.
In addition, we exploit the fact that the chordal distance is gradient dominant if $d_2(\estimGGD{}{k},\subs{t})\leq r_\mathrm{b}$ holds for some $r_\mathrm{b}<1$.
To show that this condition always holds, we first require that the initial estimate is in the ball $\mathbb{B}_{r_\mathrm{b}}$ centered around the true subspace.
Further, we assume that the ratio between the signal and noise in $W_t$ is such that the gradient descent steps in the inner loop keep the estimates close to the changing true subspace (see Figure~\ref{fig:invar}).
These conditions are formalized in the following assumption.
\begin{figure}
    \centering
    \begin{tikzpicture}
    \begin{axis}[
        axis equal,
        axis lines=none,
        xmin=-4, xmax=4,
        ymin=-4, ymax=4,
    ]
        \addplot[
            domain=0:360,
            samples=200,
            smooth,
            thick,
            myGreen
        ] ({2*cos(x)}, {2*sin(x)});
        
        \addplot[
            domain=0:360,
            samples=200,
            smooth,
            thick,
            myRed
        ] ({cos(x)}, {sin(x)});
        
        \addplot[
            domain=0:360,
            samples=200,
            smooth,
            thick,
            myBlue
        ] ({3*cos(x) + cos(30)}, {3*sin(x) + sin(30)});

        \addplot[
            domain=0:360,
            samples=200,
            smooth,
            dashed,
            thick,
            myGreen
        ] ({2*cos(x) + cos(30)}, {2*sin(x) + sin(30)});

        \node[fill=black, circle, inner sep=1pt] at (axis cs:0,0) {};
        \node[black,anchor=north west] at (axis cs:0,0) {$\subs{t-1}$};

        \node[fill=black, circle, inner sep=1pt] at (axis cs:0.866,0.5) {};
        \node[black,anchor=south west] at (axis cs:0.866,0.5) {$\subs{t}$};

        \node[fill=black, circle, inner sep=1pt] at (axis cs:-2*0.866,-2*0.5) {};
        \node[black,anchor=north east] at (axis cs:-2*0.866,-2*0.5) {$\estim{t-1}$};

        \node[fill=black, circle, inner sep=1pt] at (axis cs:-2*0.5+0.866,-2*0.866+0.5) {};
        \node[black,anchor=north] at (axis cs:-2*0.5+0.866,-2*0.866+0.5) {$\estim{t}$};

        \draw[black, dashed,-Latex] (0,0) -- (0.866,0.5);
        \node[black,anchor=south east] at (axis cs:0.5*0.866,0.25) {$c$};

        \draw[black,dashed,-Latex] (-2*0.866,-2*0.5) -- (-2*0.5+0.866,-2*0.866+0.5);

    \end{axis}
\end{tikzpicture}
    \caption{Illustration of Assumption~\ref{ass:sufficient_decrease} and 
    Lemma~\ref{lemma:grad_dom}. 
    If the estimate $\estim{t-1}$ is in the metric ball $\mathbb{B}_{r_\mathrm{b}-c}(\subs{t-1})$ (solid yellow circle), it is also in $\mathbb{B}_{r_\mathrm{b}}(\subs{t})$ (blue circle), since $\subs{t} \in \mathbb{B}_c(\subs{t-1})$ (red circle) holds by Assumption~\ref{ass:LTV-Lipschitz}.
    Furthermore, under Assumption~\ref{ass:sufficient_decrease}, the gradient descent updates in the inner loop of the \gls{GREAT} algorithm guarantee that the distance between the estimate and $\subs{t}$ reduces from $r_\mathrm{b}$ to $r_\mathrm{b}-c$, i.e., $\estim{t}\in\mathbb{B}_{r_\mathrm{b}-c}(\subs{t})$ (dashed yellow circle).
    Therefore, the estimates are always contained in the ball $\mathbb{B}_{r_\mathrm{b}}$ centered around the current true subspace.}
    \label{fig:invar}
\end{figure}
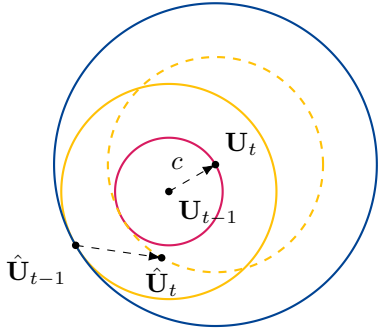
\begin{assum} \label{ass:sufficient_decrease}
The initial estimate satisfies $\estim{t_0} \in \mathbb{B}_{r_\mathrm{b}}(\subs{t_0 + 1})$ for some $c \leq r_\mathrm{b} < 1$, and there exist a step size $\alpha$, a window length $T$, and an iteration number $K$, such that, for all $t\in \mathbb{Z}_{\geq t_0}$,
    \begin{align} \label{eq:SNR_ass}
        \gamma_{r_\mathrm{b}}(\delta_t) \leq (1-\tilde{\rho})r_\mathrm{b}^2 + \frac{(1-\tilde{\rho})(c^2-2cr_\mathrm{b})}{1-\tilde{\rho}^K},
    \end{align}
    with 
    \begin{align} \label{eq:rho_tilde}
        \tilde{\rho} := 1-4(1-r_\mathrm{b}^2)(\alpha \underline{\sigma}^2 - 2\alpha^2\overline{\sigma}^4) \geq0.
    \end{align}
\end{assum}
The assumption provides a condition that depends both on the design parameters $\alpha,T,K$, and the constants $\epsilon,c,\underline{\sigma},\overline{\sigma}$ from Assumptions~\ref{ass:meas_err}-\ref{ass:PE} quantifying the properties of the data.
For fixed design parameters, inequality~\eqref{eq:SNR_ass} can be interpreted as a lower bound on the signal-to-noise ratio of the data matrix $W_t$.
The right-hand side grows with $\underline{\sigma}$, which characterizes the signal part of $W_t$.
By contrast, the left-hand side is related to the noise bound $\delta_t$ through the $\mathcal{K}_\infty$ function $\gamma_{r_\mathrm{b}}$.
\begin{remark}
    In case only one gradient descent iteration is performed in the inner loop of Algorithm~\ref{alg:GGD-exp}, i.e., $K=1$ holds, the condition~\eqref{eq:SNR_ass} can be expressed as
    \begin{align*}
        \underline{\sigma}^2 \geq \frac{\gamma_{r_\mathrm{b}}(\delta_t) + c(2r_\mathrm{b} -c)}{4\alpha r_\mathrm{b}^2(1-r_\mathrm{b}^2)} + 2\alpha\overline{\sigma}^4.
    \end{align*}
    This condition is an explicit lower bound on the signal component of $W_t$.
    Note that the bound increases if $c$ grows, as $c\leq r_\mathrm{b}$ must hold.
\hfill \QED
\end{remark}
With these assumptions in place, we conclude that the estimates always lie in the ball $\mathbb{B}_{r_\mathrm{b}}$ around the true sequence of subspaces, see Figure~\ref{fig:invar} for an illustration.
In this set, the chordal distance satisfies a gradient dominance property, as formalized below.
\begin{lemma}[Gradient dominance] \label{lemma:grad_dom}
    Let Assumptions~\ref{ass:meas_err}, \ref{ass:LTV-Lipschitz} and \ref{ass:PE} hold, and let us choose $\alpha,T,K$ such that Assumption \ref{ass:sufficient_decrease} holds.
    Then, the estimates in Algorithm~\ref{alg:GGD-exp} satisfy, for all $t\in\mathbb{Z}_{\geq t_0}$ and all $k=0,\dots,K$,
    \begin{enumerate}
        \item $\estimGGD{t}{k} \in \mathbb{B}_{r_\mathrm{b}}(\subs{t})$,
        \item $\|\mathrm{grad}~ d_{2}(\estimGGD{t}{k},\subs{t})^2\|_F^2 \geq 4(1-r_\mathrm{b}^2)d_{2}(\estimGGD{t}{k},\subs{t})^2.$
    \end{enumerate}
\end{lemma}

We are now ready to analyze how the distance between the true subspace and the estimate evolves under Algorithm~\ref{alg:GGD-exp}.
\begin{thm}[Convergence and uncertainty quantification] \label{thm:main}
    Define $\delta_t,\gamma_{r_\mathrm{b}},\tilde{\rho}$ as in~\eqref{eq:delta_def_true},~\eqref{eq:gamma_def},~\eqref{eq:rho_tilde}.
    Let Assumptions~\ref{ass:meas_err}, \ref{ass:LTV-Lipschitz}, \ref{ass:PE} hold, and choose $\alpha,T,K$ such that Assumption \ref{ass:sufficient_decrease} holds.
    Then, the output $\{\estim{t}\}_{t>t_0}$ of Algorithm~\ref{alg:GGD-exp} is such that 
    \begin{align} 
    \begin{split} \label{eq:main_thm_bound}
        d_{2}(\estim{t},\subs{t})^2 \leq & \tilde{\rho}^{Kt} d_{2}(\estim{t_0},\subs{t_0})^2 + \frac{1-\tilde{\rho}^{Kt}}{1-\tilde{\rho}}\gamma_{r_\mathrm{b}}(\|\delta\|_\infty) \\
        & + \frac{1-\tilde{\rho}^{Kt}}{1-\tilde{\rho}^K}\tilde{\rho}^K(2r_\mathrm{b}-c)c,~~ \forall t\in\mathbb{Z}_{\geq t_0+1},
    \end{split}
    \end{align}
    where $\tilde{\rho} \in (0,1)$ for all $\alpha$ satisfying $0 < \alpha < \dfrac{\underline{\sigma}^2}{2\overline{\sigma}^4}$.
    Consequently, the following ultimate bound is satisfied
    \begin{align} \label{eq:main_thm_ultimate_bound}
    \begin{split}
        \limsup_{t\to\infty} d_2(\estim{t},\subs{t})^2 = &\frac{1}{1-\tilde{\rho}}\gamma_{r_\mathrm{b}}(\|\delta\|_\infty) \\ 
        & + \frac{\tilde{\rho}^K}{1-\tilde{\rho}^K}(2r_\mathrm{b}-c)c.
    \end{split}
    \end{align}
\end{thm}
The proof of Theorem~\ref{thm:main} can be found in the Appendix, and follows by combining the single step improvement bound in Lemma~\ref{lemma:one-step-improvement-noisy} and the gradient dominance property in Lemma~\ref{lemma:grad_dom}.
Theorem~\ref{thm:main} guarantees that the subspace estimate remains in an invariant tube around the true sequence of subspaces as illustrated in Figure~\ref{fig:tube}.
The tube is characterized by the squared chordal distance between $\estim{t}$ and $\subs{t}$.
Inequality~\eqref{eq:main_thm_bound} quantifies how the tube size (also called bias) is affected by the measurement error $\epsilon$ on the data and time variation $c$ of the true subspace affecting the noise portion of the measurements in~\eqref{eq:delta_def}.
Less noise in the samples and/or slower variation of the true subspace leads to a smaller uncertainty region.
Furthermore, the ultimate bound in~\eqref{eq:main_thm_ultimate_bound} captures the asymptotic behavior of the tube size.

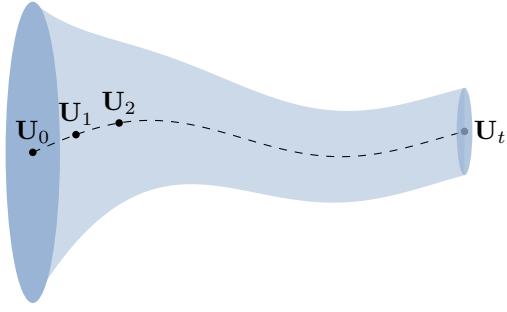
\begin{figure}
    \centering
    \begin{tikzpicture}
\def\rho{0.6}
\def\bias{0.6}
\def\dini{1.5}
    \begin{axis}[
        axis lines=middle,
        axis lines=none,
        xmin=-1, xmax=11,
        ymin=-3, ymax=3,
        domain=0:10,
        samples=100,
        smooth,
    ]

    \coordinate (Pstart) at (0,0);
    
    \draw[fill=myBlue!40, draw=none] (Pstart) ellipse [x radius=0.3*(\dini+\bias), y radius=\dini+\bias];

    \node[black,fill=black, circle, inner sep=1pt] at (Pstart) {};
    \node[black,anchor=south] at (Pstart) {$\subs{0}$};
    
    \addplot+[black,dashed,no marks] {
        0.3 * sin(deg(x/1.5)) + 
        0.2 * sin(deg(x/7)) + 
        0.1 * sin(deg(x/3))
    };

    \addplot+[black,name path=uc,no marks, draw=none] {
        0.3 * sin(deg(x/1.5)) + 
        0.2 * sin(deg(x/7)) + 
        0.1 * sin(deg(x/3)) +
        \rho^x*\dini + \bias
    };
    
    \addplot+[black,name path=lc,no marks,draw=none] {
        0.3 * sin(deg(x/1.5)) + 
        0.2 * sin(deg(x/7)) + 
        0.1 * sin(deg(x/3)) -
        \rho^x*\dini - \bias
    };

    \addplot[myBlue!20] fill between [of=uc and lc];

    \coordinate (Pend) at (10,{0.3*sin(deg(10/1.5)) + 0.2*sin(deg(10/7)) + 0.1*sin(deg(10/3))} );

    \node[black,fill=black, circle, inner sep=1pt] at (Pend) {};
    \node[black,anchor= west] at (Pend) {$\subs{t}$};
    
    \draw[fill=myBlue!40, fill opacity=0.75, draw=none] (Pend) ellipse [x radius=0.3*(\rho^10*\dini+\bias), y radius=\rho^10*\dini+\bias];

    \coordinate (Pone) at (1,{0.3*sin(deg(1/1.5)) + 0.2*sin(deg(1/7)) + 0.1*sin(deg(1/3))} );
    \node[black,fill=black, circle, inner sep=1pt] at (Pone) {};
    \node[black,anchor= south] at (Pone) {$\subs{1}$};

    \coordinate (Ptwo) at (2,{0.3*sin(deg(2/1.5)) + 0.2*sin(deg(2/7)) + 0.1*sin(deg(2/3))} );
    \node[black,fill=black, circle, inner sep=1pt] at (Ptwo) {};
    \node[black,anchor= south] at (Ptwo) {$\subs{2}$};
    
    \end{axis}
\end{tikzpicture}
    \caption{Illustration of the bound in Theorem~\ref{thm:main}. 
    The blue tube illustrates a sequence of metric balls with varying radius centered around the true sequence of subspaces (dashed line).
    The estimates from the \gls{GREAT} algorithm are guaranteed to remain in the tube.
    The evolution of the tube's radius is a function of the bound on the true subspace's rate of change $c$, the measurement error $\epsilon$, and the signal $\overline{\sigma},\underline{\sigma}$.}
    \label{fig:tube}
\end{figure}

Note that the \gls{GREAT} algorithm produces estimates with diminishing bias in the sense that $\lim_{t\to\infty} d_2(\estim{t},\subs{t}) = 0$ for $c=0$ and $\epsilon = 0$.
Furthermore, exponential convergence to the true subspace can be guaranteed in this case, which is summarized in the following corollary:
\begin{cor}[Diminishing bias] \label{cor:exp_decay}
    Let Assumptions~\ref{ass:meas_err}, \ref{ass:LTV-Lipschitz} and \ref{ass:PE} hold, and choose $\alpha,T,K$ such that Assumption \ref{ass:sufficient_decrease} holds. Furthermore, suppose that $c = 0$ and $\epsilon=0$. Then, the output $\{\estim{t}\}_{t>t_0}$ of Algorithm~\ref{alg:GGD-exp} is such that
    \begin{align*}
        d_2(\estim{t},\subs{t})^2 \leq \tilde{\rho}^{Kt} d_2(\estim{t_0},\subs{t_0})^2,\quad \forall t \in \mathbb{Z}_{\geq t_0+1}.
    \end{align*}
\end{cor}
The claim directly follows from~\eqref{eq:main_thm_bound}, the definition of $\delta_t$ in~\eqref{eq:delta_def_true}, and the fact that $\gamma_{r_\mathrm{b}} \in \mathcal{K}_\infty$.
The condition $c=0$ and $\epsilon = 0$ holds true, e.g., when \gls{LTI} systems are identified from exact measurements (c.f., Rem.~\ref{rem:Fundamental_lemma}).    

\begin{remark} \label{rem:step-size}
    The step size of the gradient descent update in Algorithm~\ref{alg:GGD-exp} is an important tuning parameter.
    Choosing its value from the interval $(0,\underline{\sigma}^2/(2\overline{\sigma}^4))$ guarantees that $\rho \in (0,1/8)$; see Figure~\ref{fig:rho} for a visualization.
    Further, this guarantees that the bound in~\eqref{eq:main_thm_bound} converges, as $1-4(1-r_\mathrm{b}^2)\rho = \tilde{\rho} \in (0,1)$ holds.
    The maximal convergence rate (minimal $\tilde{\rho}$) is achieved with $\alpha=\alpha^\mathrm{cvg} = \underline{\sigma}^2/(4\overline{\sigma}^4)$.
    On the other hand, the ultimate bound in~\eqref{eq:main_thm_ultimate_bound} is also affected by the step size $\alpha$.
    Namely, $\tilde{\rho}$ and $\gamma_r$ defined in Lemma~\ref{lemma:one-step-improvement-noisy} and Assumption~\ref{ass:sufficient_decrease} are rational functions of $\alpha$.
    In case the constants from Assumption~\ref{ass:meas_err}--\ref{ass:sufficient_decrease} are known, one can find the step size $\alpha^\mathrm{ub}$ numerically, for which the ultimate bound is minimized.
    The trade-off between these conflicting objectives is illustrated in the numerical example in Section~\ref{sec:ex1}.
\hfill\QED \end{remark}

\begin{remark} 
    Inequality~\eqref{eq:main_thm_bound} can be interpreted through the lens of nonlinear control, revealing a parallel to input-to-state inequalities~\cite{sontag1989smooth}.
    Specifically, the noise bound $\delta_t$ and the time-variation bound $c$ on the true subspace may be regarded as exogenous inputs.
    When the input is zero, the inequality ensures that the system's state decays in some norm, analogous to the squared distance $d_2(\estim{t}, \subs{t})^2$ in our setup (see Corollary~\ref{cor:exp_decay}).
    For non-zero inputs, $d_2(\estim{t}, \subs{t})^2$ remains bounded, with the bound scaling with the input size, as illustrated in~\eqref{eq:main_thm_ultimate_bound}.
\hfill \QED
\end{remark}

\begin{remark} \label{rem:forgetting_factor}
    We note that it is also possible to achieve adaptation by discounting the past data via an exponential forgetting factor $0<\gamma<1$.
    In this case, we can utilize all past data by defining the data matrix $W_t$ as
    \begin{equation*}
        W_t = [\gamma^{t-1}\traj{}{1}~\gamma^{t-2}\traj{}{2}~\dots~\gamma^0 u_t].
    \end{equation*}
    Even though the number of columns in $W_t$ grows over time, the final computations are not affected, since the matrix $W_t$ only appears as an inner product $W_tW_t^\top$, see Remark~\ref{rem:computation}. 
    This product can be computed by a rank-1 update as $W_{t+1}W_{t+1}^\top = \gamma^2 W_tW_t^\top + \traj{}{t+1}\traj{}{t+1}^\top$.
    Alternatively, one can also combine the two approaches by considering a moving data window with discounted past data.
    The theoretical analysis incorporating forgetting factors leads to a different bound in Lemma~\ref{lemma:LTV-Hankel-noisy}, which can be derived following similar arguments.
    We use a moving data window instead of a forgetting factor in this work, since this approach performed favorably in the online system identification problem considered in Section~\ref{sec:ex2}.
\hfill \QED
\end{remark}

\section{Numerical case studies} \label{sec:num-ex}
This section presents two numerical case studies.
First, we illustrate the theoretical properties of the proposed method through a synthetic example.
Second, we demonstrate how the \gls{GREAT} algorithm can be applied to online system identification as described in Section~\ref{sec:mot_ex}.
This example considers an \gls{LTV} airplane model adopted from the literature.
The MATLAB code reproducing both case studies is available online at \url{https://gitlab.ethz.ch/asasfi/ST_for_sysID}.

\subsection{Tracking a random geodesic} \label{sec:ex1}
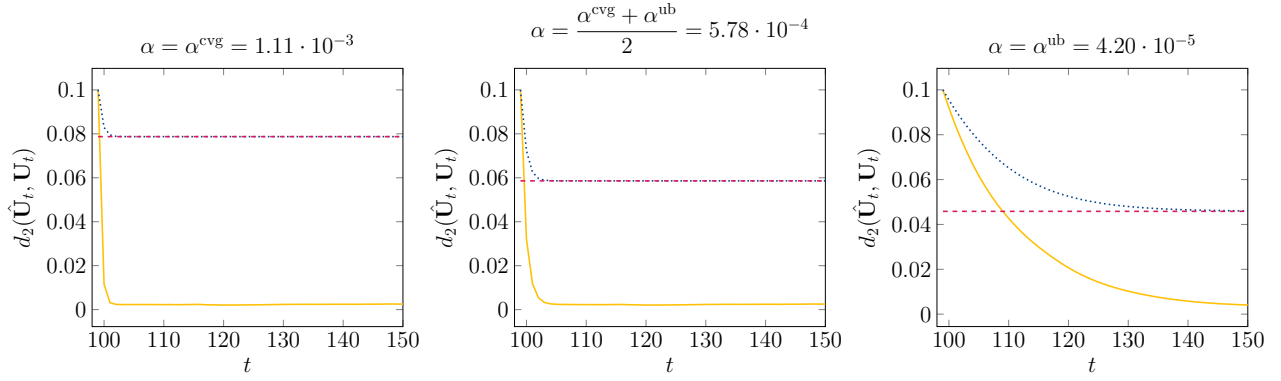
\begin{figure*}
    \centering
    \begin{tikzpicture}[scale=0.6,every node/.style={font=\Large}] 
\begin{axis}[
    title={$\alpha = \alpha^\mathrm{cvg} = 1.11\cdot10^{-3}$},
    name=plot1,
    xlabel={$t$},
    xmin=98, xmax=150,
    ylabel={$d_2(\estim{t},\subs{t})$},
yticklabel style={
        /pgf/number format/fixed,
        /pgf/number format/precision=5
},
scaled y ticks=false
]

\addplot[no marks, solid, myGreen, line width=1pt] table [x index=0,y index=1, col sep=comma] {figures/ex1_1_data.csv}; 
\addplot[no marks, myBlue, dotted, line width=1pt] table [x index=0,y index=2, col sep=comma] {figures/ex1_1_data.csv};
\addplot[no marks, dashed, myRed, line width=1pt] table [x index=0,y index=3, col sep=comma] {figures/ex1_1_data.csv}; 
\end{axis}

\begin{axis}[
    title={$\alpha = \dfrac{\alpha^{\mathrm{cvg}} + \alpha^{\mathrm{ub}}}{2} = 5.78 \cdot10^{-4}$},
    name=plot2,
    at=(plot1.right of north east), anchor=left of north west,
    xlabel={$t$},
    xmin=98, xmax=150,
    ylabel={$d_2(\estim{t},\subs{t})$},
    yticklabel style={
        /pgf/number format/fixed,
        /pgf/number format/precision=5
},
scaled y ticks=false
]
\addplot[no marks, solid, myGreen, line width=1pt] table [x index=0,y index=1, col sep=comma] {figures/ex1_2_data.csv}; 
\addplot[no marks, myBlue, dotted, line width=1pt] table [x index=0,y index=2, col sep=comma] {figures/ex1_2_data.csv}; 
\addplot[no marks, dashed, myRed, line width=1pt] table [x index=0,y index=3, col sep=comma] {figures/ex1_2_data.csv}; 
\end{axis}

\begin{axis}[
    title={$\alpha = \alpha^\mathrm{ub} = 4.20 \cdot10^{-5}$},
    name=plot3,
    at=(plot2.right of north east), anchor= left of north west,
    xlabel={$t$},
    xmin=98, xmax=150,
    ylabel={$d_2(\estim{t},\subs{t})$},
    yticklabel style={
        /pgf/number format/fixed,
        /pgf/number format/precision=5
},
scaled y ticks=false
]
\addplot[no marks, solid, myGreen, line width=1pt] table [x index=0,y index=1, col sep=comma] {figures/ex1_3_data.csv}; 
\addplot[no marks, myBlue, dotted, line width=1pt] table [x index=0,y index=2, col sep=comma] {figures/ex1_3_data.csv};
\addplot[no marks, dashed, myRed, line width=1pt] table [x index=0,y index=3, col sep=comma] {figures/ex1_3_data.csv};
\end{axis}

\end{tikzpicture}
    \caption{Illustration of the subspace tracking algorithm \gls{GREAT} with the derived theoretical bounds.    
    For different step sizes $\alpha$, the evolution of the distance between the estimates and the true subspace (yellow, solid) is shown with the bound from~\eqref{eq:main_thm_bound} (blue, dotted) and the ultimate bound from~\eqref{eq:main_thm_ultimate_bound} (red, dashed) in Theorem~\ref{thm:main}. The step size $\alpha^\mathrm{cvg}$ (left) maximizes the convergence rate, and $\alpha^\mathrm{ub}$ (right) minimizes the ultimate bound~(c.f., Remark~\ref{rem:step-size}).}
    \label{fig:ex1}
\end{figure*}
The goal of this synthetic example is to assess the conservatism of the bounds in Theorem~\ref{thm:main} and to illustrate the trade-off when selecting the step size $\alpha$, described in Remark~\ref{rem:step-size}. 
To better highlight the conservatism originating from the analysis, we construct the example such that the bounds in Assumptions~\ref{ass:meas_err}-\ref{ass:PE} are tight.
We generate a reference sequence of subspaces $\subs{t} \in \Grass{5}{3}$ for $t=0,\dots,150$ defined by a geodesic that starts from an initial subspace $\subs{0}$ and follows a tangent direction $V$.
The initial subspace $\subs{0}$ is spanned by the first $3$ unit vectors in $\mathbb{R}^5$. 
To generate a direction $V$, we first sample a matrix randomly from a standard normal distribution and then project it onto the tangent space of $\subs{0}$.
The spacing between consecutive subspaces is chosen such that Assumption~\ref{ass:LTV-Lipschitz} holds with $c=5 \cdot 10^{-5}$.
The samples are generated as $u_t = \subsST{t}\xi_t + e_t$, with $\xi_t \sim \mathcal{N}(0,I_3)$ and $e_t\sim \mathcal{N}(0,I_5)$, and $e_t$ was scaled such that $\|e_t\|_2 = \epsilon = 10^{-3}$ for all $t=1,\dots,150$ (c.f., Assumption~\ref{ass:meas_err}).
We choose the data window length to be $\windowlength = 100$, and perform $K=10$ iterations of the gradient descent method at each time $t$.
The constants with values $\underline{\sigma} = 8.49$ and $\overline{\sigma}= 11.28$ were computed, such that the inequalities in Assumption~\ref{ass:PE} are tight on the time interval $t=[100,150]$.
The initial estimate $\estim{99}$ is chosen from the set $\mathbb{B}_{r_\mathrm{b}}(\subs{100})$ with $r_\mathrm{b}=0.1$. We note that Assumption~\ref{ass:sufficient_decrease} is satisfied with the provided parameters.

We compute $\alpha^\mathrm{ub}$ (c.f., Remark~\ref{rem:step-size}) by minimizing the ultimate bound using the \textit{fmincon} function of MATLAB, and run the algorithm on the same dataset with step sizes $\alpha\in\left \{\alpha^\mathrm{cvg},~\frac{\alpha^\mathrm{cvg}+\alpha^\mathrm{ub}}{2},~\alpha^\mathrm{ub}\right \}$ to illustrate the trade-off between the convergence rate and the ultimate bound.
The distance between the resulting estimates and the true subspace, as well as the theoretical bounds are depicted in Figure~\ref{fig:ex1}.
Clearly, the fastest convergence is achieved with $\alpha^\mathrm{cvg}$ for both the theoretical bound and the true distance. 
On the other hand, the smallest ultimate bound is achieved with $\alpha^\mathrm{ub}$ at the expense of slower convergence.
In between these edge cases, the intermediate value $\alpha = 5.78\cdot10^{-4}$ provides a good trade-off, achieving both fast convergence and small ultimate bound.

\subsection{Online system identification} \label{sec:ex2}
The second example shows how the proposed method can be used for online identification of \gls{LTV} systems.
First, we compare our approach with two methods from the system identification toolbox of MATLAB that are often used as benchmarks.
Namely, we consider the \textit{N4SID} algorithm~\cite{ljung1999system,van1996subspace,van1994n4sid,verhaegen1994identification,larimore1990canonical} that estimates a state-space model using a subspace identification method at each step, and the \textit{recursiveLS} algorithm, that estimates the parameters of a multivariate ARX model online using recursive least squares technique~\cite{ljung1999system}.
Next, we focus on the subspace system model only, and compare the \gls{GREAT} algorithm with the GROUSE~\cite{balzano2010online} and \gls{PAST}~\cite{yang1995PAST} algorithms for subspace tracking.
We emphasize that only the proposed algorithm offers rigorous guarantees when the system, and correspondingly the reference subspace, is time-varying.
The identified models are then used to predict future outputs, and the prediction error is compared.
We assume that no prior system knowledge is available, and therefore, the hyperparameters of each identification method must be tuned by validation.

\subsubsection{Data generation}
We use the airplane model described in~\cite{astrom2008adaptive} to generate input-output data.
Different \gls{LTI} state-space representations of this third-order system are given under different flight conditions.
The airplane has $m=1$ input and we assume that the full state is measured, i.e., the number of outputs is $p=3$.
We simulate the behavior of the airplane when accelerating from 0.9 to 1.5 Mach number at the altitude of 35000 feet by linearly interpolating the elements of the state-space matrices between the two operating conditions.
We add a linear state-feedback $u_\mathrm{fb} = [-0.09, -0.8, 0]x$ that is stabilizing in both operating conditions, and discretize the closed-loop system leading to the \gls{LTV} state-space equations
    \begin{align*}
        x_{t+1} &= A^{\mathrm{stab}}_t x_t + B_t v_t, \\
        y_t & = x_t + \eta_t,
    \end{align*}
where the matrix $A^{\mathrm{stab}}_t$ includes the effect of the stabilizing feedback $u_\mathrm{fb}$, and $\eta_t \sim \mathcal{N}(0,0.01 I_3)$ is i.i.d. measurement error.
Note that the above equations can be transformed to the multivariate output error model structure~\cite[Sec. 4.2]{ljung1999system}, with output order $n_a=1$ and input order $n_b=1$.
We simulate the resulting discrete-time system for $1000$ steps starting from zero initial conditions and random inputs drawn from a normal distribution with unit variance.
We split the dataset into \textit{initialization}, \textit{validation}, and \textit{test} parts of equal size.
We simulate $100$ trajectories in the test part with different input and measurement error realizations to enable statistical analysis of the results.

\subsubsection{Implementation details}
We identify the restricted behavior online as
described in Section~\ref{sec:mot_ex} using the \gls{GREAT} algorithm and the GROUSE and \gls{PAST} algorithms.
We select the trajectory length to be $10~ (L=9)$, making the dimension of the embedding space $n=40$.
The dimension of the subspace is a hyperparameter, which we select $d=13$ from the range $\{11,12,13,14,15\}$, as it achieves the smallest prediction error on the validation set.
Note that $d=13$ is the correct dimension, since $m=1$ and $\sysorder=3$ for the data generating system, and $d=\sysorder +m(L+1)$ (compare equation~\eqref{eq:state-space_subspace}).
For the \gls{GREAT} and \gls{PAST} algorithms, we choose the data window length and the forgetting factor as $\windowlength = 120$ and $\beta = 0.985$ by validation from the ranges $\{30,60,90,120,150\}$ and $\{0.98,0.985,0.99,0.995,1\}$, respectively.
In practice, prior knowledge about the system can guide the selection of the hyperparameter ranges used for validation.
The algorithms are initialized as discussed in Section~\ref{sec:tracking} using the initialization data.
For improved performance, the gradient descent step on the manifold for \gls{GREAT} and GROUSE is implemented using the \textit{manopt} toolbox~\cite{manopt}, which optimizes the step size online by line-search.
The maximal number of steps is limited to $K=2$.
We predict the future outputs using the estimated behavior as described in Section~\ref{sec:mot_ex} with $T_\mathrm{ini}=T_\mathrm{fut}=5$.

The state-space model is estimated online by calling the \textit{N4SID} function at each time.
The order of the system is automatically selected as $k=3$ from the interval $\{1,2,\dots,10\}$ based on the initialization data, which is again the correct value.
To enhance adaptation, we discount the past data by a forgetting factor that is selected as $\beta~=~0.99$ from the range $\{0.98,0.985,0.99,0.995,1\}$ by validation.
We note that computationally more efficient algorithms for identifying state-space models recursively are available in the literature under the name of recursive subspace identification~\cite{verhaegen1991fast,oku2002recursive,lovera2000recursive,mercere2007convergence,mercere2008propagator}.
However, we do not consider computational aspects in this example.
In order to make predictions with the identified model, the state of the system is estimated using a Kalman filter.
For the ARX model, the output and input orders (c.f.,~\cite{ljung1999system}) are selected as $n_a=3$ and $n_b=1$, respectively, both from the interval $\{1,2,\dots,10\}$.
The rate of adaptation is controlled by the forgetting factor, which is judiciously selected as $0.985$ from the range $\{0.98,0.985,0.99,0.995,1\}$ by validation.
The initial coefficients are computed by minimizing the prediction error on the initialization data.

\subsubsection{Results}
We compare the different identification methods by their predictive performance on the test set.
The mean prediction error together with its standard deviation are computed across the $100$ trajectories in the test dataset.
At time $t$, we predict the next $T_\mathrm{fut}=5$ outputs denoted by $\hat{y}_i,~i=t+1,\dots,t+T_\mathrm{fut}$.
We assume that the outputs are available up to time $t$, while the inputs are known in the future as well.
Note that we simulate the state-space and ARX models to obtain the sequence of predictions.
The relative error is given as
\begin{align*}
    \text{relative prediction error} = \sqrt{\frac{\sum_{i=t+1}^{t+T_\mathrm{fut}} \| \hat{y}_i - y_i\|^2}{\sum_{i=t+1}^{t+T_\mathrm{fut}}  \|y_i\|^2}}.
\end{align*}

First, we compare the predictive performance of different systems models identified by the \gls{GREAT} algorithm, \textit{N4SID} and the \textit{recursiveLS} methods.
The mean relative prediction error and its standard deviation are depicted in Figure~\ref{fig:sys_id_example1}.
All the methods perform similarly both in terms of mean error and standard deviation, and successfully adapt the model online.
We note the \textit{N4SID} and \textit{recursiveLS} methods achieve adaptation by heuristically introducing a forgetting factor in the problem formulation.
Rigorous analysis for these type of methods is only available for \gls{LTI} systems.
On the other hand, the proposed method provides non-asymptotic theoretic certificates for \gls{LTV} systems as described in Section~\ref{sec:main}, while matching the performance of the other methods.
\begin{figure}
    \centering
    \begin{tikzpicture}
\begin{axis}[
    xlabel={t},
    xmin=0, xmax=330,
    ymin=0.05, ymax = 0.3,
    ylabel={Relative prediction error}
]
\addplot[name path=uc_n4sid,no marks,draw=none] table [x index=0,y index = 2, col sep = comma] {figures/ex2_1_data_rev.csv};
\addplot[name path=uc_arx,no marks,draw=none] table [x index=0,y index = 5, col sep = comma] {figures/ex2_1_data_rev.csv};
\addplot[name path=uc,no marks, draw=none] table [x index=0,y index = 8, col sep = comma] {figures/ex2_1_data_rev.csv};

\addplot[name path=lc_n4sid,no marks, draw=none] table [x index=0,y index = 3, col sep = comma] {figures/ex2_1_data_rev.csv};
\addplot[name path=lc_arx,no marks, draw=none] table [x index=0,y index = 6, col sep = comma] {figures/ex2_1_data_rev.csv};
\addplot[name path=lc,no marks, draw=none] table [x index=0,y index = 9, col sep = comma] {figures/ex2_1_data_rev.csv};

\addplot[myGreen!10] fill between [of=uc_n4sid and lc_n4sid];
\addplot[myRed!10] fill between [of=uc_arx and lc_arx];
\addplot[myBlue!10] fill between [of=uc and lc];

\addplot[no marks, myGreen, line width=1pt] table [x index=0,y index=1, col sep=comma] {figures/ex2_1_data_rev.csv}; 
\addplot[no marks, myRed, line width=1pt] table [x index=0,y index=4, col sep=comma] {figures/ex2_1_data_rev.csv}; 
\addplot[no marks, myBlue, line width=1pt] table [x index=0,y index=7, col sep=comma] {figures/ex2_1_data_rev.csv}; 

\end{axis}
\end{tikzpicture}
    \caption{Relative prediction error of the system models estimated by \textit{N4SID} recomputed online (yellow), \textit{recursiveLS} (red), and the proposed method (blue). The solid lines are the average errors across the $100$ test trajectories, and the shaded area denotes $+/-$ one standard deviation. Note that the standard deviations are very similar for the three methods, and hence, the shaded areas overlap.}
    \label{fig:sys_id_example1}
\end{figure}
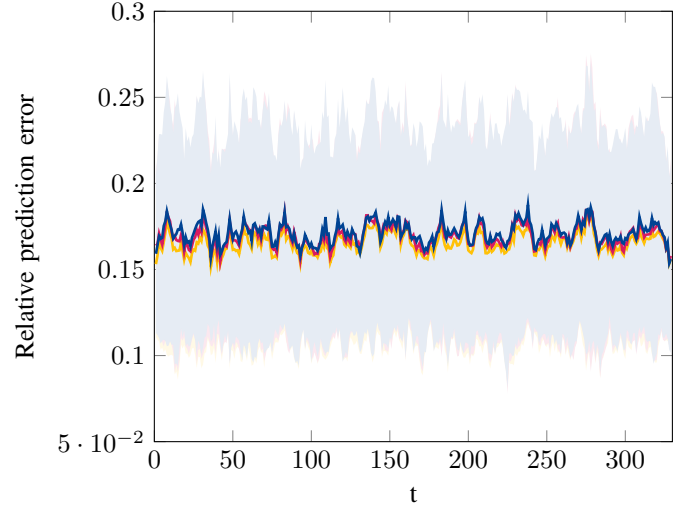
To demonstrate robustness, we introduce a large measurement error in the test data at $t = 100$.
As shown in Figure~\ref{fig:sys_id_example2}, the proposed method and \textit{N4SID} computed online maintains a significantly smaller average error and standard deviation compared to \textit{recursiveLS}.
\begin{figure}
    \centering
    \begin{tikzpicture}
\def\fault{95}
\begin{axis}[
    xlabel={t},
    xmin=0, xmax=330,
    ymin=0, ymax = 2,
    ylabel={Relative prediction error}
]
\addplot[name path=uc_n4sid,no marks,draw=none] table [x index=0,y index = 2, col sep = comma] {figures/ex2_2_data_rev.csv};
\addplot[name path=uc_arx,no marks,draw=none] table [x index=0,y index = 5, col sep = comma] {figures/ex2_2_data_rev.csv};
\addplot[name path=uc,no marks, draw=none] table [x index=0,y index = 8, col sep = comma] {figures/ex2_2_data_rev.csv};

\addplot[name path=lc_n4sid,no marks, draw=none] table [x index=0,y index = 3, col sep = comma] {figures/ex2_2_data_rev.csv};
\addplot[name path=lc_arx,no marks, draw=none] table [x index=0,y index = 6, col sep = comma] {figures/ex2_2_data_rev.csv};
\addplot[name path=lc,no marks, draw=none] table [x index=0,y index = 9, col sep = comma] {figures/ex2_2_data_rev.csv};

\addplot[myGreen!20] fill between [of=uc_n4sid and lc_n4sid];
\addplot[myRed!20] fill between [of=uc_arx and lc_arx];
\addplot[myBlue!20] fill between [of=uc and lc];

\addplot[no marks, myGreen, line width=1pt] table [x index=0,y index=1, col sep=comma] {figures/ex2_2_data_rev.csv}; 
\addplot[no marks, myRed, line width=1pt] table [x index=0,y index=4, col sep=comma] {figures/ex2_2_data_rev.csv}; 
\addplot[no marks, myBlue, line width=1pt] table [x index=0,y index=7, col sep=comma] {figures/ex2_2_data_rev.csv};

\addplot[black,dashed] coordinates {(\fault, 0) (\fault, 2)};
\end{axis}
\end{tikzpicture}
    \caption{Relative prediction error of the system models estimated by \textit{N4SID} recomputed online (yellow)}, \textit{recursiveLS} (red), and the proposed method (blue). The solid lines are the average errors across the $100$ test trajectories, and the shaded area denotes $+/-$ one standard deviation. A large measurement error occurs at $t=100$ leading to a sudden increase in the error (dashed vertical line).
    After the large measurement error, the standard deviation of the prediction error produced by the \textit{recursiveLS} method increases to $2.64$, and thus, it is not displayed entirely.
    \label{fig:sys_id_example2}
\end{figure}
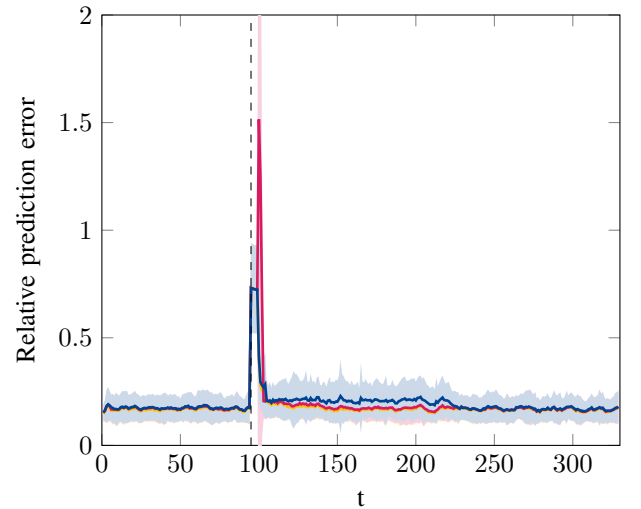

Next, we compare different subspace tracking algorithms for the identification of the subspace system model.
The evolution of relative prediction errors are displayed in Figure~\ref{fig:sys_id_example3}.
Clearly, the \gls{PAST} algorithm achieves significantly lower prediction errors than the GROUSE.
The proposed algorithm performs equally well as \gls{PAST}, while it is based on Grassmannian gradient descent that is also used in the GROUSE method.
Hence, by using theoretical tools from differential geometry, we are able to provide strong guarantees for the \gls{GREAT} algorithm even when tracking a time-varying reference.
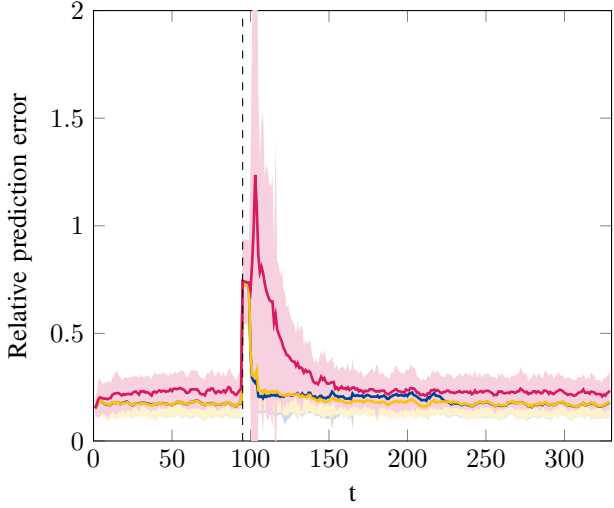
\begin{figure}
    \centering
    \begin{tikzpicture}
\def\fault{95}
\begin{axis}[
    xlabel={t},
    xmin=0, xmax=330,
    ymin=0, ymax = 2,
    ylabel={Relative prediction error}
]

\addplot[name path=uc,no marks, draw=none] table [x index=0,y index = 8, col sep = comma] {figures/ex2_2_data_rev.csv};
\addplot[name path=uc_PAST,no marks,draw=none] table [x index=0,y index = 11, col sep = comma] {figures/ex2_2_data_rev.csv};
\addplot[name path=uc_GROUSE,no marks,draw=none] table [x index=0,y index = 14, col sep = comma] {figures/ex2_2_data_rev.csv};

\addplot[name path=lc,no marks, draw=none] table [x index=0,y index = 9, col sep = comma] {figures/ex2_2_data_rev.csv};
\addplot[name path=lc_PAST,no marks, draw=none] table [x index=0,y index = 12, col sep = comma] {figures/ex2_2_data_rev.csv};
\addplot[name path=lc_GROUSE,no marks, draw=none] table [x index=0,y index = 15, col sep = comma] {figures/ex2_2_data_rev.csv};

\addplot[myBlue!20] fill between [of=uc and lc];
\addplot[myGreen!20] fill between [of=uc_PAST and lc_PAST];
\addplot[myRed!20] fill between [of=uc_GROUSE and lc_GROUSE];

\addplot[no marks, myBlue, line width=1pt] table [x index=0,y index=7, col sep=comma] {figures/ex2_2_data_rev.csv}; 
\addplot[no marks, myGreen, line width=1pt] table [x index=0,y index=10, col sep=comma] {figures/ex2_2_data_rev.csv}; 
\addplot[no marks, myRed, line width=1pt] table [x index=0,y index=13, col sep=comma] {figures/ex2_2_data_rev.csv};

\addplot[black,dashed] coordinates {(\fault, 0) (\fault, 2)};
\end{axis}
\end{tikzpicture}
    \caption{Relative prediction error of the subspace system models estimated by the \textit{\gls{GREAT}} (blue), the \textit{\gls{PAST}} (yellow), and the GROUSE (red) algorithms. The solid lines are the average errors across the $100$ test trajectories, and the shaded area denotes $+/-$ one standard deviation. A large measurement error occurs at $t=100$ leading to a sudden increase in the error (dashed vertical line).}
    \label{fig:sys_id_example3}
\end{figure}

\section{Conclusions} \label{sec:conclusion}
We introduced the \gls{GREAT} algorithm that tracks time-varying subspaces based on gradient descent on the Grassmann manifold.
Under suitable assumptions on the online data, the uncertainty in the resulting estimates is quantified by an upper bound on their distance to the true subspace.
Furthermore, we guarantee exponential convergence to a stationary reference subspace in case of noise-free data.
In contrast to the existing literature, we are able to provide non-asymptotic guarantees for the \gls{GREAT} algorithm, even when the subspace represents a dynamical system.
Hence, the proposed scheme is suitable for online identification of linear time-varying dynamical systems, which is also demonstrated in one of the numerical examples.

Future work includes developing an adaptive data-driven control framework using the online estimated subspace (c.f., Subspace Predictive Control~\cite{favoreel1999spc}). 
Furthermore, the uncertainty quantification of the estimate can be incorporated to robustify the resulting control formulation.
Another possible application of the proposed algorithm is fault detection based on the change in system behavior represented by the subspace.
Finally, extensions of the \gls{GREAT} algorithm allowing for the adaptation of the subspace dimension could be explored~\cite{padoan2025distances}.

\section{Appendix}
\appendices

\section*{Proof of Lemma~\ref{lemma:LTV-Hankel-noisy}}
\begin{proof}
    The noisy data matrix can be written as 
    \begin{align*}
        W_t = \underbrace{[\ntraj{}{t-\windowlength+1}~\dots~\ntraj{}{t}]}_{:=\bar{W}_t} + \underbrace{[e_{t-\windowlength+1}~\dots~e_t]}_{:=E_t}.
    \end{align*}
    Using Lemma~\ref{lemma:gap_metric} and the fact that $\ntraj{t-L-i+1}{t-i} \in \subs{t-i}$ for all $i = 0,\dots,\windowlength-1$, the projection error of the nominal data in $\bar{W}_t$ can be bounded as
    \begin{align*}
    \|P_{\subs{t}}^\perp & \bar{W}_t \|_F^2 = {\sum_{i=0}^{\windowlength-1} \| P_{\subs{t}}^\perp \ntraj{t-L-i+1}{t-i}\|^2_2} \\
    & \leq {\sum_{i=0}^{\windowlength - 1} d_\infty(\subs{t},\subs{t-i})^2 \|\ntraj{t-L-i+1}{t-i}\|^2_2} \\
    & \leq {\sum_{i=0}^{\windowlength - 1} d_{2}(\subs{t},\subs{t-i})^2 \|\ntraj{t-L-i+1}{t-i}\|^2_2} \\
    & \leq {\sum_{i=1}^{\windowlength-1} \left( \sum_{k=0}^{i-1} d_{2}(\subs{t-k},\subs{t-k-1}) \right)^2 \|\ntraj{t-L-i+1}{t-i}\|^2_2} \\
    & \leq {\sum_{i=0}^{\windowlength-1} \left( i c \right)^2 \|\ntraj{t-L-i+1}{t-i}\|^2_2}  = c^2 \|\bar{W}_t D\|_F^2.
    \end{align*}
    Then 
    \begin{align*}
        \|P_{\subs{t}}^\perp {W}_t\|_F & \leq \|P_{\subs{t}}^\perp \bar{W}_t\|_F + \|P_{\subs{t}}^\perp E_t\|_F \\
        & \leq c\|\bar{W}_t D\|_F + \|E_t\|_F \\
        & \leq c\left(\|W_t D\|_F + \|E_t D\|_F\right) + \|E_t\|_F \\
        & \leq c\|W_t D\|_F + \|E_t\|_F\left(c\|D\|_2 + 1\right).
    \end{align*}       
    The claim follows from $\|E_t\|_F = \sqrt{\sum_{i=0}^{\windowlength-1} \|e_{t-i}\|_2^2} \leq \epsilon\sqrt{\windowlength}$ and $\|D\|_2\leq \windowlength-1$.
    Note that $\delta_t$ is strictly monotone, in fact, linear in $c$ and $\epsilon$.
\end{proof}

\section*{Proof of Lemma~\ref{lemma:one-step-improvement-noisy}}
In order to prove Lemma~\ref{lemma:one-step-improvement-noisy}, we first show that the Riemannian gradient $\mathrm{grad}~d_{2}(\mathbf{U},\mathbf{V})^2$ is Lipschitz continuous in the sense of~\cite[Def. 10.40]{boumal2023introduction}.
\begin{lemma} \label{lemma:Lip}
    The function $\mathrm{grad}~d_{2}(\mathbf{U},\mathbf{V})^2$ on the Grassmann manifold is $L$-Lipschitz continuous in the first argument in the sense of~\cite[Def. 10.40]{boumal2023introduction} with $L=4$.
\end{lemma}
\begin{proof}
    Note that by~\cite[Cor. 10.47]{boumal2023introduction} the gradient of a twice differentiable function is $L$-Lipschitz continuous if and only if its Hessian has operator norm bounded by $L$. 
    The Riemannian Hessian of $d_{2}(\mathbf{U},\mathbf{V})^2 = \mathrm{tr}(U^\top P_{\mathbf{V}}^\perp U)$ with respect to the first argument is given as (see \cite[Example 9.49, Eq. (9.69)]{boumal2023introduction})
\begin{align*}
    \mathrm{Hess}~ d_{2}(\mathbf{U},\mathbf{V})[X] & = 2P_{\mathbf{U}}^\perp P_{\mathbf{V}}^\perp X - 2 X U^\top P_{\mathbf{V}}^\perp U.
\end{align*} 
Note that we use the Frobenius inner product, and hence, the operator norm is defined using the Frobenius norm as
\begin{align*}
    \|&\mathrm{Hess} ~ d_{2}(\mathbf{U},\mathbf{V})\| = \max_{\stackrel{X \in T_{\mathbf{U}} \Grass{n}{d}}{\|X\|_F=1}} \|\mathrm{Hess}~d_{2}(\mathbf{U},\mathbf{V}) [X]\|_F \\ 
    & = \max_{\stackrel{X \in T_{\mathbf{U}} \Grass{n}{d}}{\|X\|_F=1}} \|2P_{\mathbf{U}}^\perp P_{\mathbf{V}}^\perp X - 2XU^\top P_{\mathbf{V}}^\perp U\|_F \\ 
    & \leq \max_{\stackrel{X \in T_{\mathbf{U}} \Grass{n}{d}}{\|X\|_F=1}} \|2P_{\mathbf{U}}^\perp P_{\mathbf{V}}^\perp X\|_F + \|2XU^\top P_{\mathbf{V}}^\perp U\|_F \\
    & \leq \max_{\stackrel{X \in T_{\mathbf{U}} \Grass{n}{d}}{\|X\|_F=1}} 2\|P_{\mathbf{U}}^\perp P_{\mathbf{V}}^\perp \|_2~\|X\|_F + 2\|U^\top P_{\mathbf{V}}^\perp U\|_2~\|X\|_F \\
    & = 2\|P_{\mathbf{U}}^\perp P_{\mathbf{V}}^\perp \|_2 + 2\|U^\top P_{\mathbf{V}}^\perp U\|_2 \leq 2 + 2\|P_{\mathbf{V}}^\perp U\|_2^2 \\
    & \leq 2 + 2\|U\|^2_2 = 4.
\end{align*}
\end{proof}

In the following, we provide the proof of Lemma~\ref{lemma:one-step-improvement-noisy} using the result in Lemma~\ref{lemma:Lip}.

\begin{proof}
We first exploit the Lipschitzness of $\mathrm{grad}~d_{2}(\mathbf{U},\mathbf{V})^2$, and separate it into a nominal and a noise term in Part I.
The nominal part of the gradient contains the signal $P_{\subs{t}} W_t$ only, and it is analyzed in Part II.
We bound all terms in the gradient that are affected by the noise $P_{\subs{t}}^\perp W_t$ in Part III.

\textbf{Part I}
By Lemma~\ref{lemma:Lip}, the Riemannian gradient of the squared chordal distance is Lipschitz continuous with constant $L = 4$. Therefore, by~\cite[Cor. 10.54]{boumal2023introduction}, it holds that
    \begin{align} \label{eq:thm-one-step-improvement-1}
    \begin{split}
        d_{2}(\estimGGD{t}{k+1},\subs{t})^2 & - d_{2}(\estimGGD{t}{k},\subs{t})^2 \\
        \leq &-\alpha \langle \mathrm{grad}~ d_{2}(\estimGGD{t}{k},\subs{t})^2,\mathrm{grad}~ F_{W_t}(\estimGGD{t}{k}) \rangle_F \\
        & +2\alpha^2 \|\mathrm{grad}~ F_{W_t}(\estimGGD{t}{k})\|_F^2.
        \end{split}
    \end{align}
    The matrix $W_t$ can be written as $W_t = P_{\subs{t}}W_t  + P_{\subs{t}}^\perp W_t$. Similarly, we decompose the gradient $\mathrm{grad}~F_{W_t}(\estimGGD{t}{k})$ as $\mathrm{grad}~ F_{W_t}(\estimGGD{t}{k}) = \bar{G}_t + \gradnoise$, with
    \begin{align*}
        \bar{G}_t := & -2P_{\estimGGD{t}{k}}^\perp P_{\subs{t}} W_t W_t^\top P_{\subs{t}} \estimST{}, \\
        \gradnoise  := & -2P_{\estimGGD{t}{k}}^\perp \left(P_{\subs{t}}^\perp W_t W_t^\top P_{\subs{t}}^\perp + P_{\subs{t}}^\perp W_t W_t^\top P_{\subs{t}} \right.\\
        & \left. +P_{\subs{t}} W_t W_t^\top P_{\subs{t}}^\perp \right) \estimST{}.
    \end{align*}
    With this partitioning,~\eqref{eq:thm-one-step-improvement-1} becomes
    \begin{align} \label{eq:thm-one-step-improvement-2}
        \begin{split}
        d_{2}&(\estimGGD{t}{k+1},\subs{t})^2 - d_{2}(\estimGGD{t}{k},\subs{t})^2 \\ 
        \leq & -\alpha \left( \langle \mathrm{grad}~ d_{2}(\estimGGD{t}{k},\subs{t})^2,\bar{G}_t \rangle_F + \langle \mathrm{grad}~ d_{2}(\estimGGD{t}{k},\subs{t})^2,\gradnoise\rangle_F\right)\\
        & +2\alpha^2 \left(\|\bar{G}_t\|_F^2 + 2\langle\bar{G}_t,\gradnoise\rangle_F + \|\gradnoise\|_F^2 \right).
        \end{split}
    \end{align}
    
    \textbf{Part II} First, we analyze the nominal part of inequality~\eqref{eq:thm-one-step-improvement-2}. The gradient of the squared chordal distance on the Grassmann manifold with respect to $\estimGGD{t}{k}$ is $\mathrm{grad}~ d_{2}(\estimGGD{t}{k},\subs{t})^2 = -2 P_{\estimGGD{t}{k}}^\perp P_{\subs{t}} \estimST{}.$ Furthermore, consider the compact \gls{SVD} $\subsST{t}^\top W_t = Q_1 \Sigma Q_2^\top$. Then
    \begin{align*}
        \langle & \mathrm{grad}~ d_{2}(\estimGGD{t}{k},\subs{t})^2,\bar{G}_t \rangle_F \\
        &= 4 \mathrm{tr}(\estimST{}^\top P_{\subs{t}} P^\perp_{\estimGGD{t}{k}} P_{\subs{t}} W_t W_t^\top P_{\subs{t}} \estimST{}) \\
        & = 4 \mathrm{tr}(\subsST{t}^\top P_{\estimGGD{t}{k}} P_{\subs{t}} P^\perp_{\estimGGD{t}{k}} \subsST{t} \subsST{t}^\top W_t W_t^\top \subsST{t}) \\
        & = 4 \mathrm{tr}(Q_1^\top \subsST{t}^\top P_{\estimGGD{t}{k}} P_{\subs{t}} P^\perp_{\estimGGD{t}{k}} \subsST{t} Q_1\Sigma^2) \\
        & \geq \underline{\sigma}^2 4\mathrm{tr}(Q_1^\top \subsST{t}^\top P_{\estimGGD{t}{k}} P_{\subs{t}} P^\perp_{\estimGGD{t}{k}} \subsST{t} Q_1)\\
        & = \underline{\sigma}^2 4\mathrm{tr}( \estimST{}^\top P_{\subs{t}} P^\perp_{\estimGGD{t}{k}} P_{\subs{t}} \estimST{}) = \underline{\sigma}^2 \|\mathrm{grad}~ d_{2}(\estimGGD{t}{k},\subs{t})^2\|_F^2,
    \end{align*}
    where we used the fact that $\Sigma \geq \underline{\sigma} I_d$ by  Assumption~\ref{ass:PE}, and 
    \begin{align*}
        Q_1^\top \subsST{t}^\top & P_{\estimGGD{t}{k}} P_{\subs{t}} P^\perp_{\estimGGD{t}{k}} \subsST{t} Q_1\\
        & = Q_1^\top \left ( \subsST{t}^\top P_{\estimGGD{t}{k}} \subsST{t} - \subsST{t}^\top P_{\estimGGD{t}{k}} P_{\subs{t}} P_{\estimGGD{t}{k}} \subsST{t} \right) Q_1 \\
        & = Q_1^\top \left( \subsST{t}^\top P_{\estimGGD{t}{k}} I P_{\estimGGD{t}{k}} \subsST{t}  - \subsST{t}^\top P_{\estimGGD{t}{k}} P_{\subs{t}} P_{\estimGGD{t}{k}} \subsST{t} \right)Q_1 \\
        & = Q_1^\top \subsST{t}^\top P_{\estimGGD{t}{k}} P^\perp_{\subs{t}}P_{\estimGGD{t}{k}} \subsST{t} Q_1 \\
        & = Q_1^\top \subsST{t}^\top P_{\estimGGD{t}{k}} P^\perp_{\subs{t}} P^\perp_{\subs{t}} P_{\estimGGD{t}{k}} \subsST{t} Q_1  \geq 0,
    \end{align*}
    due to the properties of projection matrices. 
    Similarly, the following upper bound holds
    \begin{align*}
        \|\bar{G}_t\|_F^2 & = 4 \mathrm{tr}(\estimST{}^\top P_{\subs{t}} W_t W_t^\top P_{\subs{t}} P^\perp_{\estimGGD{t}{k}} P_{\subs{t}} W_t W_t^\top P_{\subs{t}} \estimST{}) \\ 
        & \leq  \overline{\sigma}^2 4 \mathrm{tr}(Q_1^\top \subsST{t}^\top P_{\estimGGD{t}{k}} P_{\subs{t}} W_t W_t^\top P_{\subs{t}} P^\perp_{\estimGGD{t}{k}} \subsST{t} Q_1) \\
        & = \overline{\sigma}^2 4 \mathrm{tr}(\subsST{t}^\top P^\perp_{\estimGGD{t}{k}} P_{\subs{t}} P_{\estimGGD{t}{k}} \subsST{t} \subsST{t}^\top W_t W_t^\top \subsST{t}) \\
        & \leq \overline{\sigma}^4 4 \mathrm{tr}(\subsST{t}^\top P^\perp_{\estimGGD{t}{k}} P_{\subs{t}} P_{\estimGGD{t}{k}} \subsST{t}) \\
        &= \overline{\sigma}^4 \|\mathrm{grad}~ d_{2}(\estimGGD{t}{k},\subs{t})^2\|_F^2.
    \end{align*}
    
    \textbf{Part III} Now, we bound the terms containing $\gradnoise$. The Cauchy-Schwarz inequality yields
    \begin{align*}
        |\langle  \mathrm{grad}~ d_{2}(\estimGGD{t}{k},\subs{t})^2,\gradnoise \rangle_F| \leq & \|\mathrm{grad}~ d_{2}(\estimGGD{t}{k},\subs{t})^2\|_F \|\gradnoise\|_F, \\
        |\langle \bar{G}_t,\gradnoise\rangle_F| \leq & \|\bar{G}_t\|_F \|\gradnoise\|_F.
    \end{align*}   
    Furthermore,
    \begin{align*}
        \|\mathrm{grad}~ d_{2}(\estimGGD{t}{k},\subs{t})^2\|_F & = \|2 P_{\estimGGD{t}{k}}^\perp P_{\subs{t}} \estimST{}\|_F \leq \|2 P_{\estimGGD{t}{k}}^\perp P_{\subs{t}} \|_F \leq 2r, \\
        \|\bar{G}_t\|_F & \leq \overline{\sigma}^2 \|\mathrm{grad}~ d_{2}(\estimGGD{t}{k},\subs{t})^2\|_F \leq 2r\overline{\sigma}^2.
    \end{align*}
    Combining the results of Part II and III with~\eqref{eq:thm-one-step-improvement-2} yields
    \begin{align*}
    \begin{split}
        d_{2}&(\estimGGD{t}{k+1},\subs{t})^2 - d_{2}(\estimGGD{t}{k},\subs{t})^2 \\ 
        \leq & \left( -\alpha \underline{\sigma}^2 + 2\alpha^2\overline{\sigma}^4 \right) \|\mathrm{grad}~ d_{2}(\estimGGD{t}{k},\subs{t})^2\|_F^2 \\
        & + 2\alpha r \|\gradnoise\|_F + 8\alpha^2 r \overline{\sigma}^2 \|\gradnoise\|_F + 2\alpha^2 \|\gradnoise\|_F^2.
    \end{split}
    \end{align*}
    Finally, the triangle inequality and Lemma~\ref{lemma:LTV-Hankel-noisy} lead to
    \begin{align*}
        \|\gradnoise\|_F \leq & 2\|P_{\subs{t}}^\perp W_t W_t^\top P_{\subs{t}}^\perp + P_{\subs{t}}^\perp W_t W_t^\top P_{\subs{t}} 
        \\
        & +P_{\subs{t}} W_t W_t^\top P_{\subs{t}}^\perp\|_F \\
        \leq & 2\|P_{\subs{t}}^\perp W_t W_t^\top P_{\subs{t}}^\perp\|_F + 4\|P_{\subs{t}} W_t W_t^\top P_{\subs{t}}^\perp\|_F \\
        \leq & 2 \|P_{\subs{t}}^\perp W_t\|_F^2 + 4 \|P_{\subs{t}} W_t\|_2~ \|P_{\subs{t}}^\perp W_t \|_F  \\
        \leq & 2 \delta_t^2 + 4 \overline{\sigma} \delta_t.
    \end{align*}
    Together these yield the right-hand side bound in Lemma~\ref{lemma:one-step-improvement-noisy}.
\end{proof}

\section*{Proof of Lemma~\ref{lemma:grad_dom}}
\begin{proof}
    For any $\mathbf{U},\tilde{\mathbf{U}} \in \Grass{n}{d}$ such that $\tilde{\mathbf{U}} \in \mathbb{B}_{r_\mathrm{b}}(\mathbf{U})$, the squared chordal distance satisfies the gradient dominance property 
\begin{align} \label{eq:main_thm_GD}
\begin{split}
    \|\mathrm{grad}~ d_{2}(\mathbf{U}, \tilde{\mathbf{U}})^2\|_F^2 & =  4\mathrm{tr}(P_{\mathbf{U}}P_{\tilde{\mathbf{U}}}P_{\mathbf{U}}^\perp P_{\tilde{\mathbf{U}}}) \\
    & = 4\mathrm{tr}(\tilde{U}^\top U U^\top \tilde{U} \tilde{U}^\top P_{\mathbf{U}}^\perp \tilde{U}) \\
    & \geq 4 \min_i \cos^2(\theta_i) \mathrm{tr}(\tilde{U}^\top P_{\mathbf{U}}^\perp \tilde{U}) \\
    & = 4 \min_i (1 - \sin^2(\theta_i)) d_{2}(\mathbf{U}, \tilde{\mathbf{U}})^2 \\
    & \geq 4(1-r_\mathrm{b}^2)d_{2}(\mathbf{U}, \tilde{\mathbf{U}})^2,
    \end{split}
\end{align}
where $\theta_i$, $i=1,\dots,d$ are the principal angles between $\mathbf{U}$ and $\tilde{\mathbf{U}}$. Therefore, it is sufficient to show that $\estimGGD{t}{k} \in \mathbb{B}_{r_\mathrm{b}}(\subs{t})$ holds for all $t\geq \windowlength$ and $k=0,\dots,K$.

First, we show by induction that the set $\mathbb{B}_{r_\mathrm{b}}(\subs{t})$ is forward invariant under the gradient descent update~\eqref{eq:GGD_update_rule}. For a fixed $t\geq \windowlength$, assume that $\estimGGD{t}{k} \in \mathbb{B}_{r_\mathrm{b}}(\subs{t})$ for some $k=0,\dots,K-1$.
Note that with the choice of $\alpha$ in Algorithm~\ref{alg:GGD-exp}, we have $\rho >0$, and hence, $\tilde{\rho} <1$. 
Therefore, we can combine the results of Lemma~\ref{lemma:one-step-improvement-noisy} with $r = r_\mathrm{b}$, and inequality~\eqref{eq:main_thm_GD} yielding
\begin{align*}
    d_{2}(\estimGGD{t}{k+1},\subs{t})^2 \leq & \tilde{\rho} d_2(\estimGGD{t}{k},\subs{t})^2 +\gamma_{r_\mathrm{b}}(\delta_t) \\
    \leq & \tilde{\rho} r_\mathrm{b}^2 +\gamma_{r_\mathrm{b}}(\delta_t).
\end{align*}
Furthermore, Assumption~\ref{ass:sufficient_decrease} implies
\begin{align*}
    \gamma_{r_\mathrm{b}}(\delta_t) \leq (1-\tilde{\rho})r_\mathrm{b}^2,    
\end{align*}
and therefore, $\estimGGD{t}{k+1} \in \mathbb{B}_{r_\mathrm{b}}(\subs{t})$ holds.

Next, we show that for any $t\geq \windowlength$, $\estim{t-1} = \estimGGD{t}{0} \in \mathbb{B}_{r_\mathrm{b}}(\subs{t})$ implies $\estim{t} = \estimGGD{t}{K} \in \mathbb{B}_{r_\mathrm{b}}(\subs{t+1})$, as illustrated in Figure~\ref{fig:invar}.
Due to the invariance of $\mathbb{B}_{r_\mathrm{b}}(\subs{t})$ under the gradient descent update~\eqref{eq:GGD_update_rule}, $\estimGGD{t}{k} \in \mathbb{B}_{r_\mathrm{b}}(\subs{t})$ holds at time $t-1$ for all $k=0,\dots,K$, and hence, inequality~\eqref{eq:main_thm_GD} applies.
We apply the bound in Lemma~\ref{lemma:one-step-improvement-noisy} recursively to get
\begin{align}
\begin{split} \label{eq:lemma_4_proof2}
    d_{2}(\estimGGD{t}{K},\subs{t})^2 \leq & \tilde{\rho}^K d_{2}(\estimGGD{t}{0},\subs{t})^2 +\sum_{k=0}^{K-1} \tilde{\rho}^k \gamma_{r_\mathrm{b}}(\delta_t) \\
    \leq & \tilde{\rho}^K r_\mathrm{b}^2 + \frac{1-\tilde{\rho}^K}{1-\tilde{\rho}} \gamma_{r_\mathrm{b}}(\delta_t) \leq (r_\mathrm{b}-c)^2,
\end{split}
\end{align}
where we used Assumption~\ref{ass:sufficient_decrease} in the last step. Furthermore, due to the triangle inequality and Assumption~\ref{ass:LTV-Lipschitz}, we have
\begin{align*}
    d_{2}(\estimGGD{t+1}{K},\subs{t+1}) \leq d_{2}(\estimGGD{t}{K},\subs{t}) + d_{2}(\subs{t},\subs{t+1}) \leq r_\mathrm{b},
\end{align*}
which completes the proof.
\end{proof}

\section*{Proof of Theorem~\ref{thm:main}}
\begin{proof}
From the proof of Lemma~\ref{lemma:grad_dom}, we have
\begin{align*}
    d_{2}(\estim{t},\subs{t})^2 \leq \tilde{\rho}^K d_{2}(\estim{t-1},\subs{t})^2 + \frac{1-\tilde{\rho}^K}{1-\tilde{\rho}} \gamma_{r_\mathrm{b}}(\delta_t).
\end{align*}
Therefore, the following bound holds for any $t\geq \windowlength$
\begin{align*}
    d_{2}(\estim{t},\subs{t})^2 \leq & \tilde{\rho}^K d_{2}(\estim{t-1},\subs{t})^2 + \frac{1-\tilde{\rho}^K}{1-\tilde{\rho}} \gamma_{r_\mathrm{b}}(\|\delta\|_\infty) \\
    \leq & \tilde{\rho}^K (d_{2}(\estim{t-1},\subs{t-1}) + c)^2 \\ 
    & + \frac{1-\tilde{\rho}^K}{1-\tilde{\rho}} \gamma_{r_\mathrm{b}}(\|\delta\|_\infty) \\
    \leq & \tilde{\rho}^K d_{2}(\estim{t-1},\subs{t-1})^2 + \tilde{\rho}^K(2r_\mathrm{b}- c)c \\
    & + \frac{1-\tilde{\rho}^K}{1-\tilde{\rho}} \gamma_{r_\mathrm{b}}(\|\delta\|_\infty),
\end{align*}
where we used the facts that $d_2(\estim{t-1},\subs{t-1}) \leq r_\mathrm{b}-c$ from inequality~\eqref{eq:lemma_4_proof2}, and $\gamma_{r_\mathrm{b}}(\delta_t) \leq \gamma_{r_\mathrm{b}}(\|\delta\|_\infty)$ for any $t \geq \windowlength$.
We apply the above inequality recursively leading to
\begin{align*}
    d_{2}(\estim{t},&\subs{t})^2 \leq \tilde{\rho}^{Kt} d_{2}(\estim{0},\subs{0})^2 \\
    & + \sum_{\tau=0}^{t-1} (\tilde{\rho}^{K})^\tau \left( \frac{1-\tilde{\rho}^K}{1-\tilde{\rho}} \gamma_{r_\mathrm{b}}(\|\delta\|_\infty) + \tilde{\rho}^K (2r_\mathrm{b}-c)c \right).
\end{align*}
Substituting the solution of the geometric series completes the proof.
\end{proof}

\section*{References}
\bibliographystyle{ieeetr}        
\bibliography{Literature}

\begin{IEEEbiography}[{\includegraphics[width=1in,height=1.25in,trim={0 5mm 0 0},clip,keepaspectratio]{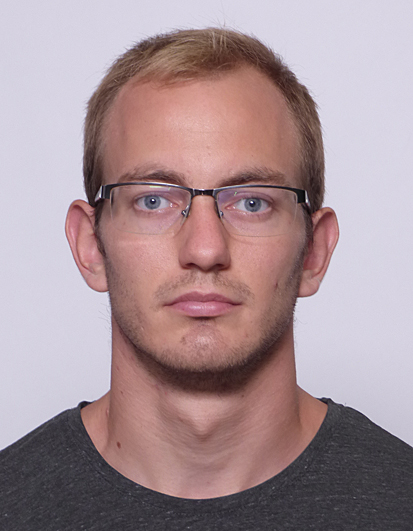}}]{Andr\'as Sasfi}
is a Ph.D. student with the Automatic Control Laboratory at ETH Z\"urich.
He received his Bachelor degree in Mechanical Engineering from the Budapest University of Technology and Economics, Hungary, in 2019.
He received his Master degree also in Mechanical Engineering from ETH Z\"urich in 2022.
His research interests include online system identification and data-driven control in the behavioral setting.
\end{IEEEbiography}

\begin{IEEEbiography}[{\includegraphics[width=1in,height=1.25in,trim={20cm 0 20cm 0},clip,keepaspectratio]{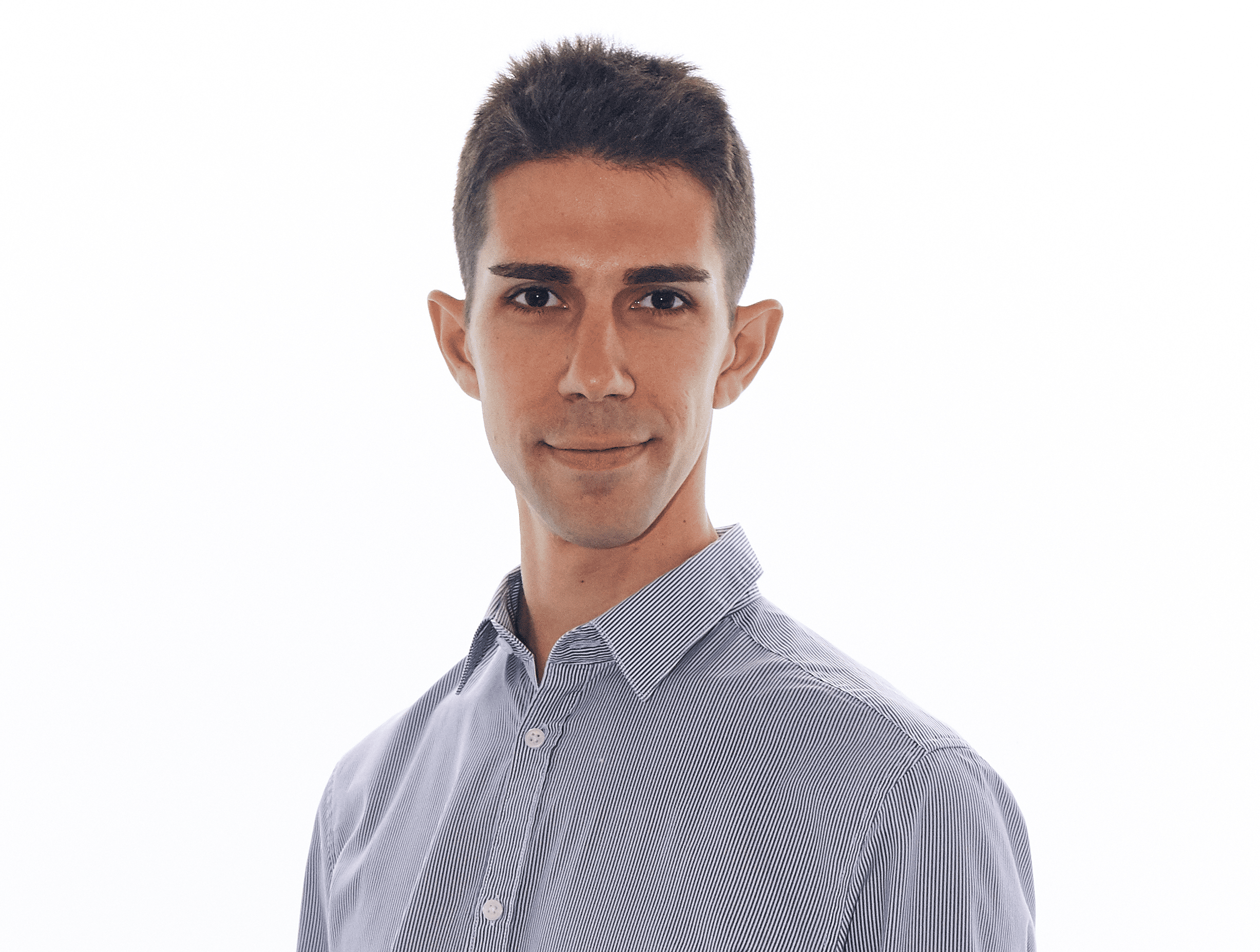}}]{Alberto Padoan}
(M’10) was born in Ferrara, Italy, in 1989. He received the Laurea Magistrale degree (M.Sc. equivalent) cum laude in Automation Engineering from the University of Padua in 2013 and the Ph.D. degree in Control Theory from Imperial College London in 2018, where his dissertation received the IET Control \& Automation Doctoral Dissertation Prize (2018) and the Eryl Cadwallader Davies Prize (2020). From 2017 to 2021, he was a Research Associate in the Control Group at the University of Cambridge and a member of Sidney Sussex College. In 2021, he joined the Automatic Control Laboratory at ETH Zurich and the NCCR Automation, first as a Postdoctoral Scholar and later as a Senior Scientist. 

Alberto is currently an Assistant Professor in the Department of Electrical and Computer Engineering at the University of British Columbia. His research focuses on the modeling, analysis, and control of complex dynamical systems, with special emphasis on biological and cyber-physical applications.
\end{IEEEbiography}

\begin{IEEEbiography}[{\includegraphics[width=1in,height=1.25in,clip,keepaspectratio]{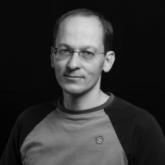}}]{{I}van Markovsky}
is an ICREA professor at the International Centre for Numerical Methods in Engineering, Barcelona. He received his Ph.D. degree in Electrical Engineering from the Katholieke Universiteit Leuven in February 2005. From 2006 to 2012 he was an Assistant Professor at the School of Electronics and Computer Science of the University of Southampton and from 2012 to 2022 an Associate Professor at the Vrije Universiteit Brussel. He is a recipient of an ERC starting grant "Structured low-rank approximation: Theory, algorithms, and applications" 2010--2015, Householder Prize honorable mention 2008, and research mandate by the Vrije Universiteit Brussel research council 2012--2022. His main research interests are computational methods for system theory, identification, and data-driven control in the behavioral setting.
\end{IEEEbiography}

\begin{IEEEbiography}[{\includegraphics[width=1in,height=1.25in,clip,keepaspectratio]{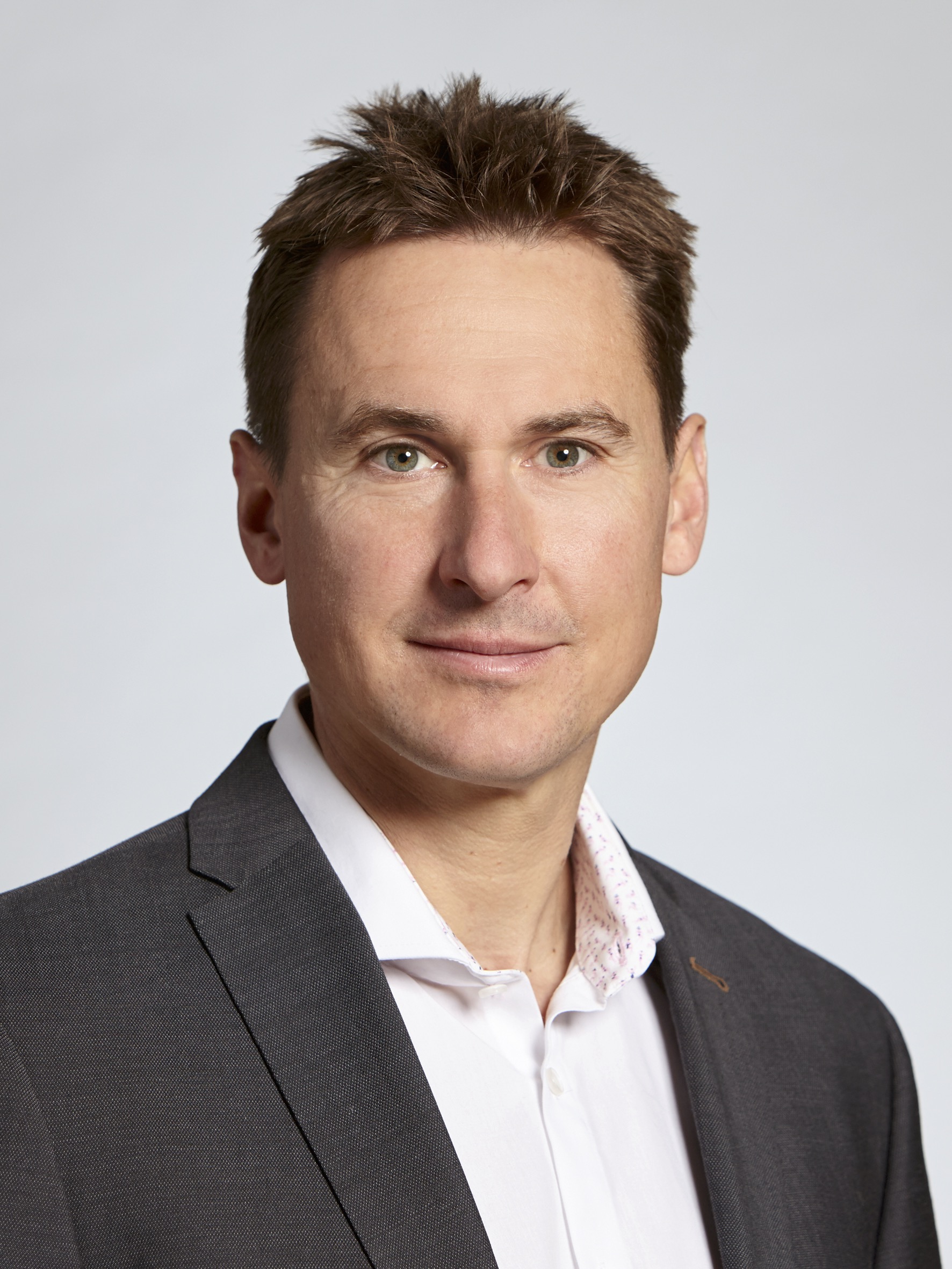}}]{Florian D\"orfler}
is a Professor at the Automatic Control Laboratory at ETH Z\"urich. He received his Ph.D. degree in Mechanical Engineering from the University of California at Santa Barbara in 2013, and a Diplom degree in Engineering Cybernetics from the University of Stuttgart in 2008. From 2013 to 2014 he was an Assistant Professor at the University of California Los Angeles. He has been serving as the Associate Head of the ETH Z\"urich Department of Information Technology and Electrical Engineering from 2021 until 2022. His research interests are centered around automatic control, system theory, optimization, and learning. His particular foci are on network systems, data-driven settings, and applications to power systems. He is a recipient of the distinguished young research awards by IFAC (Manfred Thoma Medal 2020) and EUCA (European Control Award 2020). He and his team received best paper distinctions in the top venues of control, machine learning, power systems, power electronics, circuits and systems. They were recipient of the 2011 O. Hugo Schuck Best Paper Award, the 2012-2014 Automatica Best Paper Award, the 2016 IEEE Circuits and Systems Guillemin-Cauer Best Paper Award, the 2022 IEEE Transactions on Power Electronics Prize Paper Award, the 2024 Control Systems Magazine Outstanding Paper Award, and multiple Best PhD thesis awards at UC Santa Barbara and ETH Z\"urich. They were further winners or finalists for Best Student Paper awards at the European Control Conference (2013, 2019), the American Control Conference (2010,2016,2024), the Conference on Decision and Control (2020), the PES General Meeting (2020), the PES PowerTech Conference (2017), the International Conference on Intelligent Transportation Systems (2021), the IEEE CSS Swiss Chapter Young Author Best Journal Paper Award (2022,2024), the IFAC Conference on Nonlinear Model Predictive Control (2024), and NeurIPS Oral (2024). He is currently serving on the council of the European Control Association and as a senior editor of Automatica.
\end{IEEEbiography}

\end{document}